\documentclass[11pt,oneside,english,reqno]{amsart}
\usepackage[T1]{fontenc}
\usepackage{units}
\usepackage{amsthm}
\usepackage{amstext}
\usepackage{amssymb}
\usepackage{undertilde}
\usepackage{esint}

\makeatletter
\numberwithin{equation}{section}
\theoremstyle{plain}
\newtheorem{thm}{\protect\theoremname}[section]
  \theoremstyle{definition}
  \newtheorem{defn}[thm]{\protect\definitionname}
  \theoremstyle{plain}
  \newtheorem{lem}[thm]{\protect\lemmaname}
  \theoremstyle{remark}
  \newtheorem{rem}[thm]{\protect\remarkname}

\usepackage{pdfsync} 

\makeatother

\usepackage{babel}
  \providecommand{\definitionname}{Definition}
  \providecommand{\lemmaname}{Lemma}
  \providecommand{\remarkname}{Remark}
\providecommand{\theoremname}{Theorem}

\begin{document}

\global\long\def\body{\mathcal{B}}
\global\long\def\part{\mathcal{P}}
\global\long\def\surface{\mathcal{S}}

\global\long\def\stress{X}
\global\long\def\strain{\chi}
\global\long\def\strech{\varepsilon}

\global\long\def\cbnd{d}
\global\long\def\cochain{X}

\global\long\def\Lip{\mathfrak{L}}
\global\long\def\SS#1{\mathcal{L}_{\mathfrak{s}}\left(#1\right)}
\global\long\def\Limb{\Lip_{\mathrm{Em}}}
\global\long\def\Lmap{\mathcal{F}}

\global\long\def\Emb#1{\text{Emb}\left(#1\right)}

\global\long\def\mol#1{\Phi_{#1}*}

\global\long\def\reals{\mathbb{R}}
\global\long\def\too{\longrightarrow}
\global\long\def\Vs{V}
\global\long\def\we{\wedge}

\global\long\def\Vr{\Vs_{r}}
\global\long\def\Vrd{\Vs^{r}}

\global\long\def\spt{\mathrm{spt}}

\global\long\def\imag#1{\mathrm{image}\left(#1\right)}

\global\long\def\emb{\varphi}

\global\long\def\emr{g}

\global\long\def\refi{\mathfrak{S}}

\global\long\def\sform#1#2{D^{#1}\left(#2\right)}

\global\long\def\currents#1#2{D_{#1}\left(#2\right)}

\global\long\def\Fmass#1{M\left(#1\right)}

\global\long\def\Flmass#1#2{M_{#1}\left(#2\right)}

\global\long\def\Fnormal#1{N\left(#1\right)}

\global\long\def\Fcomass#1{\left\Vert #1\right\Vert _{0}}

\global\long\def\var#1{\mathrm{Var}\left(#1\right)}

\global\long\def\ball#1#2{B\left(#1,#2\right)}

\global\long\def\mtb{\Gamma}

\global\long\def\dns#1#2{d\left(#1,#2\right)}

\global\long\def\compact{K}

\global\long\def\norm#1#2{\left\Vert #1\right\Vert _{#2}}

\global\long\def\cell{\sigma}
\global\long\def\simp{\sigma}
\global\long\def\oset{U}
\global\long\def\pspace{S}

\global\long\def\chain{\mathcal{A}}
\global\long\def\bnd{\partial}

\global\long\def\mass#1{|#1|}

\global\long\def\rest{\raisebox{0.4pt}{\,\mbox{\ensuremath{\llcorner}}\,}}
\global\long\def\irest{\raisebox{0.4pt}{\mbox{\,\ensuremath{\lrcorner}\,}}}

\global\long\def\ess{\mathrm{ess}}

\global\long\def\virv{v}
\global\long\def\mvirv{\xi}

\global\long\def\conf{\kappa}
\global\long\def\confs{\mathcal{Q}}
\global\long\def\gconf{\kappa}

\global\long\def\virvs{W_{\conf}}
\global\long\def\mvirvs{W}

\global\long\def\D{\mathfrak{D}}

\global\long\def\form{\phi}

\global\long\def\Fmass#1{M\left(#1\right)}

\global\long\def\measure#1{\mu_{#1}}

\global\long\def\Fnormal#1{N\left(#1\right)}

\global\long\def\Fflat#1#2{F_{#1}\left(#2\right)}

\global\long\def\Fsharp#1#2{S_{#1}\left(#2\right)}

\global\long\def\lusb{L}

\global\long\def\wcbd{\widetilde{\cbnd}}

\global\long\def\fform#1{D_{#1}}

\global\long\def\strech{\varepsilon}

\global\long\def\ivirv{\chi}

\global\long\def\stress{\tau}

\global\long\def\ti{\mathcal{I}}

\global\long\def\endtime{\varepsilon}

\global\long\def\motion{m}

\global\long\def\Lie{\mathcal{L}}

\global\long\def\Reyo{\mathcal{R}}

\global\long\def\drivt#1{\frac{\partial#1}{\partial t}}

\global\long\def\eps{\varepsilon}

\global\long\def\compm{K_{\motion}}

\global\long\def\bodym{\body_{\motion}}

\global\long\def\ve{\hat{v}}

\newcommand{\note}[1]{**<[#1]>**}

\title{On the Role of Sharp Chains in the Transport Theorem}

\author{L. Falach and R. Segev}
\begin{abstract}
A generalized transport theorem for convecting irregular domains is
presented in the setting of Federer's geometric measure theory. A
prototypical $r$-dimensional domain is viewed as a flat $r$-chain
of finite mass in an open set of an $n$-dimensional Euclidean space.
The evolution of such a generalized domain in time is assumed to be
in accordance to a bi-Lipschitz type map. The induced curve is shown
to be continuous with respect to the flat norm and differential with
respect to the sharp norm on currents in $\reals^{n}$. A time dependent
property is naturally assigned to the evolving region via the action
of an $r$-cochain on the current associated with the domain. Applying
a representation theorem for cochains the properties are shown to
be locally represented by an $r$-form. Using these notions a generalized
transport theorm is presented.
\end{abstract}

\maketitle

\section{Introduction}

Reynolds' transport theorem \cite{Reynolds}, offers a general form
for the formulation of basic conservation laws in continuum mechanics
and in particular in fluid dynamics. The traditional formulation of
Reynolds theorem (or Leibniz-Reynolds theorem) deals with the time
derivative of the integral of $\omega(t)$, a time dependent scalar
field, over a time evolving spatial region $\part(t)$ in a Euclidean
physical space. The region $\part(t)$ is assumed to be the image
of a domain under a smooth motion, where the domain is assumed to
be sufficiently regular such that the classical divergence theorem
is applicable (for example a Lipschitz domain). The transport theorem
states that 
\[
\frac{d}{dt}\left(\int_{\part(t)}\omega(t)d\lusb^{n}\right)\mid_{t=\tau}=\int_{\part(\tau)}\frac{\partial\omega}{\partial t}\mid_{t=\tau}d\lusb^{n}+\int_{\bnd\part(\tau)}\omega(\tau)V\cdot\nu dH^{n-1},
\]
with $\nu$, the unit exterior normal to the boundary $\bnd\part$
and $V$ the velocity associated with the smooth motion. It is noted
that the proof of the Reynolds' transport theorem is attributed by
Truesdell and Toupin \cite[p.~347]{Truesdell1960} to Spielrein (1916).

In the study of deforming thin films or evolving phase boundary a
transport relation for surface integrals is of interest. Such a theorem
is usually refereed to as \emph{Surface Transport Theorem}. It seems
that the basic notions of surface transport theorem in the setting
of continuum mechanics were first introduced in \cite{Gurtin_Murdoch1975}
with the introduction of the \emph{surface divergence operator}. 
Betounes \cite{Betounes1986}, examined the kinematics of an $r$-dimentional
submanifold embedded in an $n$-dimentional semi-Riemannian manifold.
Betounes's formulation brings to light the strong dependence of the
formulation of the surface transport theorem on the availability of
the mean curvature normal. In \cite{Gurtin1989} Gurtin et al. formulated
a surface transport theorem for moving interfaces while additional
study and applications were presented in \cite{Gurtin1988,Angenent1989,Gurtin1990,Gurtin2000,Gurtin2002},
to name a few. 

In the aforementioned versions of the transport theorem, the regularity
of the evolving domain is tacitly assumed. The inclusion of irregular
domains, where such notions as the exterior normal and mean curvature
normal not applicable, in the formulation of the transport theorem
has been presented recently in \cite{Seguin2013,Seguin2013a}. Seguin
\& Fried construct a generalized transport theorem using the setting
of Harrison's theory of differential chains (see \cite{Harrison2012,Harrison2013}).
The proposed formulation allows for singularities to evolve in the
domains, \textit{e.g.}, the domains may develop holes, split into
pieces and the fractal dimensions associated with the domain considered
may change. As their formulation of a transport theorem relies of
a dual relation between the domains and the properties considered
the resulting representation for theses properties is fairly regular.

In \cite{Falach2014}, a transport theorem is presented in the setting
of general manifolds. The domain of integration considered is viewed
as a de-Rham current of compact support thus including highly irregular
domains. The domain is assumed to evolve under a smooth map and integration
of a given property in the classical theory is replaced with the action
of the evolving current on a smooth differential form.

In the present work we wish to present a version of the transport
theorem in the setting of Federer's geometric measure theory. A generalized
domain, or a control volume, is viewed as a flat $r$-chain of finite
mass $T$. The current $T$ is assumed to evolve under the action
of $\conf$, a time dependent bi-Lipschitz homeomorphisms. With $\ti\subset\reals$
representing a time interval the control volume at time $t\in\ti$
is represented by the flat $r$-chain $\conf_{t\#}T$, given by the
pushforward of $T$ by $\conf_{t}=\conf(t)$. A considerable portion
of this work is dedicated to the study of the properties of the induced
curve $t\mapsto\conf_{t\#}T$ which is shown to be continuous with
respect to the flat topology of currents and differentiable with respect
to the sharp topology of currents. The integration of a given property
over the domain is generalized to the action of an $r$-cochain on
the current. 

It is observed that Lipschitz continuity arises naturally as a characteristic
of the proposed setting both is the resulting representation of properties,
by sharp forms, and in the regularity of the motion, independently.

\section{Notation and Preliminaries\label{sec:Notation-and-Preliminaries}}

In this section we review some of the fundamental concepts of the
theory of currents in an $n$-dimensional Euclidean space. Throughout,
the notation is in the same spirit of \cite[Chapter 4]{Federer1969}. 

Let $\oset$ be an open set in $\reals^{n}$, the notation $\D^{r}\left(\oset\right)$
is used for the vector space of smooth, compactly supported real valued
differential $r$-forms defined on $\oset$. The vector space $\D^{r}\left(\oset\right)$
is endowed with a family of semi-norms $\norm{\cdot}{i,\compact}$
such that for a compact $\compact\subset\oset$ and $i\in\mathbb{N}$,
\[
\norm{\form}{i,\compact}=\sup\left\{ \norm{D^{j}\form(x)}{}\mid x\in\compact,\;0\leq j\leq i\right\} .
\]
(Here, the norm $\norm{D^{j}\form(x)}{}$ is induced by some norm
on tensors in $\reals^{n}$.) This family of seminorms endows $\D^{r}\left(\oset\right)$
with a locally convex topology. For $\form\in\D^{r}(\oset)$ we use
$\cbnd\form$ to denote the \textit{exterior derivative }\textit{\emph{of}}
$\form$, an element of $\D^{r+1}(\oset)$. A linear functional $T:\D^{r}(U)\to\reals$
continuous with respect to the topology on $\D^{r}\left(\oset\right)$
is referred to as an $r$-\emph{dimensional}\textit{\emph{ }}\textit{de
Rham current} in $\oset$. The collection of all $r$-dimensional
currents defined on $\oset$ forms the vector space $\D_{r}(U)=\left[\D^{r}(\oset)\right]^{*}$,
the dual vector space of $\D^{r}(\oset)$. Let $T\in\D_{r}(\oset)$
with $r\geq1$ then, $\bnd T$, the\textit{ boundary} of $T$, is
the element of $\D_{r-1}(U)$ defined by 
\begin{equation}
\bnd T(\form)=T(\cbnd\form),\quad\text{for all}\quad\form\in\D^{r-1}(\oset).
\end{equation}
Thus, we have the boundary operator $\bnd=\cbnd^{*}$, the adjoint
operator of the exterior derivative. The support of a current $T\in\D_{r}(\oset)$
is defined by 
\begin{equation}
\spt\left(T\right)=\oset\setminus{\textstyle \bigcup}W,
\end{equation}
where each $W$ is an open subset of $U$ such that $T(\phi)=0$ for
all $\phi\in\D^{r}(\oset)$, with $\spt(\phi)\subset W$. Generally
speaking, the support of a current $T\in\D_{r}\left(\oset\right)$
need not be compact, however, in this work all the currents considered
will be of compact support.

The inner product in $\reals^{n}$ induces an inner product in $\bigwedge_{r}\reals^{n}$,
the vector space of $r$-vector in $\reals^{n}$, and $\mass{\xi}$
will denote the resulting norm of an $r$-vector $\xi$. An $r$-vector
$\xi$ may be written by $\xi=\sum_{\lambda\in\Lambda(r,n)}\xi^{\lambda}e_{\lambda}$
where $\Lambda(r,n)$ is the collection of increasing maps from $\{1,\dots,r\}$
to $\left\{ 1,\dots,n\right\} $, and $\left\{ e_{\lambda}\right\} $
is the standard basis for $\bigwedge_{r}\reals^{n}$ defined by $e_{\lambda}=e_{\lambda(1)}\wedge\dots\wedge e_{\lambda(r)}$.
Thus, with the above notation, $\mass{\xi}=\sqrt{\left\langle \xi,\xi\right\rangle }=\sqrt{\left(\xi^{\lambda}\right)^{2}}.$
Given $\form\in\D^{r}(\oset)$, for every $x\in\oset$, $\form(x)$
is an $r$-covector. For $\form(x)$, as well as any other covector,
one defines 
\begin{equation}
\norm{\form(x)}0=\sup\left\{ \form(x)(\xi)\mid\mass{\xi}\leq1,\,\,\xi\text{ is a simple }r\text{-vector}\right\} .
\end{equation}
For a compact subset $\compact\subset\oset$, define the $\compact$-Lipschitz
constant of $\form\in\D^{r}(\oset)$ by 
\begin{equation}
\Lip_{\phi,\compact}=\sup_{x,\, y\in\compact}\frac{\norm{\form(y)-\form(x)}0}{\mass{y-x}}.\label{eq:Lip_form_K}
\end{equation}
In addition to the topology of test functions, three additional topologies
will be examined on $\D^{r}(\oset)$ each of which is induced by a
corresponding family of semi-norms. For $\compact\subset\oset$ the
\textit{$\compact$-comass semi-norm}\textit{\emph{ of}}\textit{ $\form\in\D^{r}\left(\oset\right)$}
is defined by 
\begin{equation}
\Flmass{\compact}{\form}=\ess\sup_{x\in\compact}\left\{ \norm{\form(x)}0\right\} ,\label{eq:Mass_form}
\end{equation}
the $K$-\textit{flat semi-norm} on $\D^{r}\left(\oset\right)$ is
defined by 
\begin{equation}
\Fflat{\compact}{\form}=\ess\sup_{x\in\compact}\left\{ \norm{\form(x)}0,\norm{\cbnd\form(x)}0\right\} ,\label{eq:Flat_norm_forms}
\end{equation}
and the $\compact$-\emph{sharp semi-norm} on $\D^{r}\left(\oset\right)$
is defined by 
\begin{equation}
\Fsharp{\compact}{\form}=\sup\left\{ \sup_{x\in\compact}\norm{\form(x)}0,(r+1)\Lip_{\form,\compact}\right\} .\label{eq:Sharp_norm_forms}
\end{equation}
The factor $(r+1)$ in the above definition is introduced so that
\begin{equation}
\Fflat{\compact}{\form}\leq\Fsharp{\compact}{\form},\quad\text{for all }\form\in\D^{r}(\oset).\label{eq:F<S}
\end{equation}
(The essential supremum is used above in spite of the fact that we
consider smooth functions because we are going to apply below these
definitions to essentially bounded functions.) Later on, when the
$\ess\sup$ in each of the above terms is evaluated over $\oset$,
we shall write $\Fmass{\form},\,\Fflat{}{\form},\,\Fsharp{}{\form}$
for the mass, flat and Sharp norms, respectively.

For $T\in\D_{r}(\oset)$, the \textit{mass }\textit{\emph{of}}\textit{
$T$} is dually defined by 
\begin{equation}
\Fmass T=\sup\left\{ T\left(\form\right)\mid\phi\in\D^{r}\left(\oset\right),\,\,\Fmass{\form}\leq1\right\} ,\label{eq:Mass_current}
\end{equation}
and 
\[
\Flmass{\compact}T=\sup\left\{ T\left(\form\right)\mid\phi\in\D^{r}\left(\oset\right),\,\,\Flmass{\compact}{\form}\leq1\right\} .
\]

An $r$-dimensional current $T$ is said to be \textit{represented
by integration} if there exists a Radon measure $\measure T$ and
an $r$-vector valued, $\measure T$-measurable function, $\overrightarrow{T}$,
with $\mass{\overrightarrow{T}(x)}=1$ for $\measure T$-almost all
$x\in\oset$, such that 
\begin{equation}
T\left(\form\right)=\int_{\oset}\form(\overrightarrow{T})d\measure T,\quad\text{for all}\quad\form\in\D^{r}(\oset).\label{eq:current_by_integration}
\end{equation}
A sufficient condition for an $r$-dimensional current, $T$, to be
represented by integration is that $T$ is a current of locally finite
mass, \textit{i.e.}, $\Flmass{\compact}T<\infty$ for all compact
subsets $\compact\subset\oset$. An $r$-current $T$ of compact support
is said to be \textit{a normal current} if both $T$ and $\bnd T$
are represented by integration. The notion of normal currents leads
to the following definition 
\begin{equation}
\Fnormal T=\Fmass T+\Fmass{\bnd T},\label{eq:N_norm}
\end{equation}
and clearly, every $T\in\D_{r}(\oset)$ such that $\Fnormal T<\infty$
is a normal $r$-current. The vector space of all $r$-dimensional
normal currents in $\oset$ is denoted by $N_{r}\left(\oset\right)$
and for a compact set $\compact$ of $\oset$, 
\begin{equation}
N_{r,\compact}\left(\oset\right)=N_{r}(\oset)\cap\left\{ T\mid\spt\left(T\right)\subset\compact\right\} .\label{eq:N_km}
\end{equation}

The \textit{$K$-flat norm }on $\D_{r}\left(\oset\right)$ is given
by 
\begin{equation}
\Fflat{\compact}T=\sup\left\{ T\left(\form\right)\mid F_{\compact}\left(\form\right)\leq1\right\} .\label{eq:Flat_norm_currents}
\end{equation}
It follows naturally from the foregoing definition that 
\begin{equation}
\Fflat{\compact}{\bnd T}\leq\Fflat{\compact}T.\label{eq:F(bndT)<F(T)}
\end{equation}
Note that if $T\in\D_{r}\left(\oset\right)$ such that $\Fflat{\compact}T<\infty$,
then, $\spt(T)\subset\compact$. For a given compact subset $\compact\subset\oset$,
the set $F_{r,\compact}(U)$ is defined as the $F_{\compact}$-closure
of $N_{r,\compact}(U)$ in $\D_{r}(U)$. In addition, set 
\begin{equation}
F_{r}(\oset)=\bigcup_{\compact}F_{r,\compact}(U),\label{eq:F_m(U)}
\end{equation}
 where the union is taken over all compact subsets $\compact$ of
$\oset$. An element in $F_{r}(\oset)$ is referred to as a \textit{flat
$r$-chain in $\oset$}. 

For $T\in F_{r,\compact}(\oset)$ it can be shown that $F_{\compact}(T)$
is given by 
\begin{equation}
F_{\compact}(T)=\inf\left\{ \Fmass{T-\bnd S}+\Fmass S\mid S\in\D_{r+1}(\oset),\,\spt(S)\subset\compact\right\} ,\label{eq:Flat_norm_Whitney}
\end{equation}
and by taking $S=0$ it follows that 
\begin{equation}
F_{\compact}(T)\leq\Fmass T.\label{eq:F<M}
\end{equation}
In addition, any element $T\in F_{r,\compact}\left(\oset\right)$
may be represented by $T=R+\bnd S$ where $R\in\D_{r}(\oset)$, $S\in\D_{r+1}(\oset)$,
such that $\spt(R)\subset\compact$, $\spt(S)\subset\compact$, and
\begin{equation}
F_{\compact}(T)=\Fmass R+\Fmass S,\label{eq:T=00003DR+bndS}
\end{equation}
so that $R$ and $S$ are of finite mass. By Equation (\ref{eq:F(bndT)<F(T)})
we note that the boundary of a flat $r$-chain is a flat $(r-1)$-chain. 

The following representation theorem for flat chains is given in \cite[Section 4.1.18]{Federer1969}.
Let $T$ be a flat $r$-chain in $\oset$, then, $T$ is represented
by
\begin{equation}
T=\lusb^{n}\wedge\eta+\bnd\left(\lusb^{n}\wedge\xi\right),\label{eq:flat_chain_representation}
\end{equation}
with $\eta$ an $\lusb^{n}\rest\oset$-summable $r$-vector field
and $\xi$ an $\lusb^{n}\rest\oset$-summable $\left(r+1\right)$-vector
field $\xi$. Here, $\lusb^{n}\rest\oset$ denotes the restriction
of the $n$-dimensional Lebesgue measure to $U$ and for any $p$-current
$T$ and any $r$-vector field $\eta$, the $(p+r)$-current $T\we\eta$
is given by
\[
T\we\eta(\psi)=T(\psi\rest\eta).
\]
For $\form\in\D_{r}\left(\oset\right)$, the action $T(\form)$ is
given therefore by 
\begin{equation}
T(\form)=\int_{U}(\form(\eta)+d\form(\xi))d\lusb^{n}.\label{eq:flat_chain_repr_a}
\end{equation}
A real valued linear functional $\cochain$ defined on $F_{r}\left(\oset\right)$
is said to be a \emph{flat $r$-cochain in $\oset$} if there exists
a number $c<\infty$ such that for any $\compact\subset\oset$ compact
subset 
\[
\cochain\left(T\right)\leq c\Fflat{\compact}T,\quad\text{for all }T\in\Fflat{r,\compact}{\oset}.
\]
The infimum of all bounds $c$ is the norm of $\cochain$. An $r$-form
$\omega$ in $\oset$ is said to be a \textit{flat $r$-form in} $\oset$
if $\omega$ and $\cbnd\omega$, taken in the distributional sense,
are $\lusb^{n}$-measurable and essentially bounded (see \cite[p.~38]{Heinonen2000}).
That is, a measurable $r$-form $\omega$ is a flat $r$-form in $\oset$
if and only if $\Fflat{}{\omega}<\infty$. An important result, \textit{Wolfe's
representation theorem} (see \cite[ch.~7]{Whitney1957}, \cite[sec.~4.1.19]{Federer1969}),
states that each flat $r$-cochains $\cochain$ can be isometrically
identified with a flat $r$-form in $\oset$. 

The notion of sharp chains was initially introduced in Whitney's classical
monograph \cite{Whitney1957}. In the following we present a construction
in the spirit of the formulation of the theory of flat chains in \cite{Federer1969}
(who does not consider sharp chains). 

The $\compact$-\emph{sharp norm} on $\D_{r}\left(\oset\right)$ is
given by 
\begin{equation}
\Fsharp{\compact}T=\sup\left\{ T\left(\form\right)\mid\Fsharp{\compact}{\form}\leq1\right\} .\label{eq:sahrp_norm_currents}
\end{equation}
For a given compact subset $\compact\subset\oset$, the set $\Fsharp{r,\compact}{\oset}$
is defined as the $S_{\compact}$-closure of $N_{r,\compact}(U)$
in $\D_{r}(U)$. In addition, we set 
\begin{equation}
S_{r}(\oset)=\bigcup_{\compact}S_{r,\compact}(U),\label{eq:S_r(U)}
\end{equation}
where the union is taken over all compact subsets $\compact$ of $\oset$.
An element in $S_{r}(\oset)$ is referred to as a \textit{sharp $r$-chain
in $\oset$}. Note that as $F_{r,\compact}(U)\subset S_{r,\compact}(U)$
and the set $N_{r,\compact}\left(\oset\right)$ is a dense set in
both $F_{r,\compact}\left(\oset\right)$ and $S_{r,\compact}\left(\oset\right)$,
it follows that every sharp $r$-chain may be viewed a the limit,
in the sharp topology, of a sequence of flat $r$-chains. A representation
theorem for general sharp chains is beyond the scope of this work,
however, for sharp chains of finite mass such a representation theorem
may be found in \cite[Chapter.~XI]{Whitney1957}. 

Let $\cochain$ be a real valued linear functional defined on $\Fsharp r{\oset}$,
then, $\cochain$ is said to be a\emph{ sharp $r$-cochain} \emph{in
$\oset$} provided there exists a number $b<\infty$ such that for
any $\compact\subset\oset$ compact subset 
\[
\cochain\left(T\right)<b\Fsharp{\compact}T,\quad\text{for all}\; T\in\Fsharp{r,\compact}{\oset}.
\]
The infimum of all bounds $b$ is the norm of $\cochain$.  An $r$-form
$\omega$ in $\oset$ is said to be a \textit{sharp $r$-form in}
$\oset$ if the coefficients of $\omega$ are bounded and Lipschitz
continuous, \textit{i.e.}, $\omega$ is a sharp $r$-from if and only
if 
\[
\Fsharp{}{\omega}=\sup\left\{ \sup_{x\in\oset}\norm{\omega(x)}0,(r+1)\Lip_{\omega}\right\} <\infty.
\]

\begin{thm}
Let $\cochain$ be a sharp $r$-cochain in $\reals^{n}$. Then, $\cochain$
can be isometrically identified with $\fform{\cochain}$, a sharp
$r$-form in $\reals^{n}$ . That is, for any flat chain given by
$T=\lusb^{n}\wedge\eta+\bnd\left(\lusb^{n}\wedge\xi\right)$ 
\[
\cochain\left(T\right)=\int\left[\fform{\cochain}\left(\eta\right)+\cbnd\fform{\cochain}\left(\xi\right)\right]d\lusb^{n}.
\]
 For a sharp $r$-chain, $T$, as flat chains are dense in the space
of sharp chains we have $T=\lim_{i\to\infty}^{S}T_{i}$ and 
\[
\cochain\left(T\right)=\lim_{i\to\infty}\cochain\left(T_{i}\right).
\]

\end{thm}
For the proof see \cite[Section V.10]{Whitney1957}

Let $\oset\subset\reals^{n}$, $V\subset\reals^{m}$ be open sets
and let $T\in\D_{k}(\oset)$, $S\in\D_{l}(V)$. Then, the Cartesian
product of $T$ and $S$ is an element of $\D_{k+l}\left(\oset\times V\right)$
denoted by $T\times S$ and defined as follows. Let $\omega\in\D^{k+l}(U\times V)$
be given by 
\[
\omega=\sum_{\mass{\alpha}=k,\mass{\beta}=l}\omega_{\alpha\beta}(x,y)dx^{\alpha}\wedge dy^{\beta},
\]
where $x\in U$, $y\in V$, and $\alpha,\;\beta$ are multi-indices.
Then, 
\begin{equation}
T\times S(\omega)=T\left(\sum_{\mass{\alpha}=k}S\left(\sum_{\mass{\beta}=l}\omega_{\alpha\beta}(x,y)dy^{\beta}\right)dx^{\alpha}\right).\label{eq:curtesian_currents_def}
\end{equation}
For the properties of the Cartesian products of currents see \cite[Section 2.3]{Giaquinta1998}.
Let$U\subset\reals$, then, the line segment $[a,b]\subset U$ defines
naturally an element of $\D_{1}\left(U\right)$ such that for $\form\in\D_{1}\left(V\right)=C_{0}^{\infty}\left(V\right)$the
action $[a,b](\form)$ is $[a,b](\form)=\int_{a}^{b}\form(t)d\lusb_{t}^{1}$.
Consider the case of $T\in\D_{r}\left(U\right)$, then, a given $\omega\in\D^{r+1}\left(V\times\oset\right)$
may be split into a horizontal and a vertical component in the form
$\omega=\omega_{H}+\omega_{V}$ by 
\begin{equation}
\omega_{H}(t,x)=\sum_{\mass{\alpha}=r}\omega_{H\alpha}(t,x)dt\wedge dx^{\alpha},\quad\omega_{V}(t,x)=\sum_{\mass{\alpha}=r+1}\omega_{V\alpha}(t,x)dx^{\alpha},
\end{equation}
where, $e_{t}$ is the pre-dual of $dt$. It is observed that 
\[
\sum_{\mass{\alpha}=r}\omega_{H\alpha}(t,x)dx^{\alpha}=\omega\rest e_{t}.
\]
Applying Equation (\ref{eq:curtesian_currents_def}) to $\left([a,b]\times T\right)(\omega)$,
\begin{equation}
\left([a,b]\times T\right)(\omega)=\int_{a}^{b}T\left(\omega_{H}\right)d\lusb_{t}^{1}=\int_{a}^{b}T\left(e_{t}\irest\omega\right)d\lusb_{t}^{1}=\int_{a}^{b}e_{t}\wedge T\left(\omega\right)d\lusb_{t}^{1}.\label{eq:product_comput}
\end{equation}
Generally speaking, the Cartesian product of two flat chains is not
a flat chain. However the Cartesian product of a flat chain and a
normal current is a flat chain (see \cite[Sec. 4.1.12]{Federer1969}.

\section{Lipschitz maps\label{sec:Lipschitz-maps}}

In this section we briefly review some of the relevant properties
of Lipschitz mappings. From the point of view of kinematics, Lipschitz
mappings will be used to model the evolution of body-like regions
in space. 

A map $\Lmap:\oset\to V$ from an open set $\oset\subset\reals^{n}$
to an open set $V\subset\reals^{m}$, is said to be a\textit{ (globally)
Lipschitz map} if there exists a number $c<\infty$ such that $\mass{\Lmap(x)-\Lmap(y)}\leq c\mass{x-y}$
for all $x,\, y\in\oset$. The \textit{Lipschitz constant }of $\Lmap$
is defined by 
\begin{equation}
\Lip_{\Lmap}=\sup_{x,y\in\oset}\frac{|\Lmap(y)-\Lmap(x)|}{|y-x|}.\label{eq:Lipschitz-constant-1}
\end{equation}
The map $\Lmap:\oset\to V$ is said to be \textit{locally Lipschitz}
if for every $x\in\oset$ there is a neighborhood $\oset_{x}\subset\oset$
of $x$ such that the restricted map $\Lmap\mid_{\oset_{x}}$ is a
Lipschitz map. For a locally Lipschitz map, $\Lmap:\oset\to\reals^{m}$,
defined on the open set $\oset\subset\reals^{n}$ and a compact subset
$\compact\subset\oset$, the restricted map $\Lmap\mid_{\compact}$
is globally Lipschitz in the sense that $\Lip_{\Lmap,\compact}$,
the $\compact$-Lipschitz constant of the map $\Lmap\mid_{\compact}$,
is given by 
\begin{equation}
\Lip_{\Lmap,\compact}=\sup_{x,y\in\compact}\frac{\mass{\Lmap(x)-\Lmap(y)}}{\mass{x-y}}.
\end{equation}

The vector space of locally Lipschitz mappings from the open set $\oset\subset\reals^{n}$
to $\reals^{m}$ is denoted by $\Lip\left(\oset,\reals^{m}\right)$.
For a compact subset $\compact\subset\oset$, define the semi-norm
\begin{equation}
\norm{\Lmap}{\Lip,\compact}=\max\left\{ \norm{\Lmap\mid_{\compact}}{\infty},\Lip_{\Lmap,\compact}\right\} ,\label{eq:Lipschit_semi_norm}
\end{equation}
on $\Lip\left(\oset,\reals^{m}\right)$, where, 
\begin{equation}
\|\Lmap\mid_{\compact}\|_{\infty}=\sup_{x\in\compact}\mass{\Lmap(x)}.
\end{equation}
The vector space $\Lip\left(\oset,\reals^{m}\right)$ is endowed with
the strong Lipschitz topology (see \cite{Fukui2005} for the definition
on Riemannian manifolds). It is the analogue of Whitney's topology
(strong topology) for the space of differentiable mappings between
open sets (see \cite[p.~35]{Hirsch}) and is defined as follows. 
\begin{defn}
\label{def:Lipschitz_strong_topology}Given $\Lmap\in\Lip\left(\oset,\reals^{m}\right)$,
for some index set $\Lambda$, let $\mathcal{O}=\left\{ O_{\lambda}\right\} _{\lambda\in\Lambda}$
be an open, locally finite cover of $\oset\subset\reals^{n}$, and
$\mathcal{K}=\left\{ \compact_{\lambda}\right\} _{\lambda\in\Lambda}$
a family of compact subsets covering $\oset$ such that $\compact_{\lambda}\subset O_{\lambda}$
and $\delta=\left\{ \delta_{\lambda}\right\} _{\lambda\in\Lambda}$
a family of positive numbers. A neighborhood $B^{\Lip}\left(\Lmap,\mathcal{O},\delta,\mathcal{K}\right)$
of $\Lmap$ in the strong topology is defined as the collection of
all $g\in\Lip\left(\oset,\reals^{m}\right)$ such that $\norm{\Lmap-g}{\Lip,\compact_{\lambda}}<\delta_{\lambda}$,
\emph{i.e.}, 
\begin{equation}
B^{\Lip}\left(\Lmap,\mathcal{O},\mathcal{K},\delta\right)=\left\{ g\in\Lip\left(\oset,\reals^{n}\right)\mid\:\norm{\Lmap-g}{\Lip,\compact_{\lambda}}<\delta_{\lambda},\,\lambda\in\Lambda\right\} .\label{eq:Lipschitz_strong_topology}
\end{equation}
For an illustrative description for the strong topology in the case
of $C^{0}$-functions, see \cite[p.~59]{Hirsch}. 
\end{defn}
The following lemma shows the strong character of convergence in the
strong topology. For the proof in the case of differentiable mappings,
see \cite[p.~27]{Michor1980} or \cite[p.~43]{Golubitsky1973}.
\begin{lem}
\label{lem:compact_strong_topology}Let $\left\{ \Lmap_{\alpha}\right\} _{\alpha=1}^{\infty}$
be a sequence in $\Lip\left(\oset,\reals^{m}\right)$. Then, the sequence
converges to $\Lmap\in\Lip\left(\oset,\reals^{m}\right)$ in the strong
Lipschitz topology, if and only if there exists a compact subset $\compact\subset\oset$
such that $\Lmap_{\alpha}$ equals $\Lmap$ on $\oset\backslash\compact$
for all but finitely many $\alpha$'s and $\Lmap_{\alpha}\mid_{\compact}$
converges to $\Lmap\mid_{\compact}$ uniformly. 
\end{lem}
A map $\emb:\oset\too V$, for open sets $\oset\subset\reals^{n}$,
$V\subset\reals^{m}$ and $m\geq n$, is said to be a \textit{bi-Lipschitz
}map if there are numbers $0<c\leq d<\infty$, such that 
\begin{equation}
c\leq\frac{|\emb(x)-\emb(y)|}{\mass{x-y}}\leq d,\qquad\text{for all}\quad x,\, y\in\oset,\; x\not=y.\label{eq:bi-Lipschitz_constant-1}
\end{equation}
(See \cite[p.~78]{Heinonen2000} for further discussion.) Setting
$L=\max\left\{ \frac{1}{c},d\right\} $, it follows that 
\begin{equation}
\frac{1}{L}\leq\frac{|\emb(x)-\emb(y)|}{\mass{x-y}}\leq L,\qquad\text{for all}\quad x,\, y\in\oset,\; x\not=y,
\end{equation}
and in such a case $\emb$ is said to be $L$-bi-Lipschitz. (See \cite[p.~78]{Heinonen2000}
for further discussion.) 
\begin{defn}
The map $\Lmap:\oset\to V$, where $\oset\subset\reals^{n}$ and $V\subset\reals^{m}$
are open sets such that $m\geq n$, is a \textit{Lipschitz immersion}
if for every $x\in\oset$ there is a neighborhood $U_{x}\subset\oset$
of $x$ such that $\Lmap\mid_{U_{x}}$ is a bi-Lipschitz map, \textit{i.e.},
there are $0<c_{x}\leq d_{x}<\infty$, and 
\begin{equation}
c_{x}\leq\frac{|\emb(y)-\emb(z)|}{\mass{y-z}}\leq d_{x},\quad\text{for all}\quad y,\, z\in U_{x},\; y\not=z.
\end{equation}

\end{defn}

\begin{defn}
A Lipschitz map $\emb:\oset\to V$ is said to be a\textit{ Lipschitz
embedding} if it is a Lipschitz immersion and a homeomorphism of $\oset$
onto $\emb(\oset)$.
\end{defn}
The following theorems pertaining to the set of Lipschitz immersions
and Lipschitz embeddings are given in \cite{Fukui2005} for the setting
of Lipschitz manifolds. Their proofs are analogous to the case of
differentiable mappings as in \cite[p.~36--38]{Hirsch}. 
\begin{thm}
The set of Lipschitz immersions is an open subset of $\Lip\left(\oset,\reals^{m}\right)$
with respect to the strong Lipschitz topology\label{thm:Immersion_open}.

\end{thm}

\begin{thm}
The set $\Limb\left(\oset,\reals^{m}\right)$ is open in $\Lip\left(\oset,\reals^{m}\right)$
with respect to the strong Lipschitz topology \label{thm:Lipschitz_embedding_open_set}.

\end{thm}
In the following, a smooth embedding $\emb$ will be an element of
the set 
\begin{equation}
\Emb{\oset,\reals^{m}}=C^{\infty}\left(\oset,\reals^{m}\right)\cap\Limb\left(\oset,\reals^{m}\right),\label{eq:DefineEmb}
\end{equation}
that is, a Lipschitz embedding whose components are smooth.

\section{The Image of Currents and Homotopy\label{sec:fushforward_homotopy}}

Let $\oset\subset\reals^{n}$ and $V\subset\reals^{m}$ be open sets,
, and let $f:\oset\to V$ be a map of class $C^{\infty}$. We recall
that for any $\omega\in\D^{r}(V)$, the pullback of $\omega$ by $f$,
is the $r$-form in $\oset$ denoted by $f^{\#}(\omega)$ such that
\begin{equation}
\begin{split}f^{\#}(\omega)(x)\left(v_{1}\wedge\dots\wedge v_{r}\right) & =\omega\left(f(x)\right)\left(Df(x)v_{1}\wedge\dots\wedge Df(x)v_{r}\right),\\
 & =\omega\left(f(x)\right)\left(\bigwedge_{r}Df(x)\left(v_{1}\wedge\dots\wedge v_{r}\right)\right),
\end{split}
\label{eq:pullback_form_smooth}
\end{equation}
for any collection of vectors $v_{1},\dots,v_{r}\in\reals^{n}$. Note
that $f^{\#}\left(\omega\right)$ need not be an element of $\D^{r}(\oset)$.
For example, let $f:\oset\to\reals^{n}$ be the inclusion then for
$\omega\in\D^{r}\left(\reals^{n}\right)$ such that $\spt(\omega)\cap\oset\not=\varnothing$
and $\spt(\omega)\cap\left(\reals^{n}\backslash\oset\right)\not=\varnothing$
then $f^{\#}\left(\omega\right)$ is not compactly supported in $\oset$
an thus $f^{\#}\left(\omega\right)\not\in\D^{r}(\oset)$. 

Let $T\in\D_{r}\left(\oset\right)$ such that $f\mid_{\spt\left(T\right)}$
is a proper map. The \emph{pushforward} of $T$ by $f$ is denoted
by $f_{\#}(T)\in\D_{r}(V)$ and is defined by 
\begin{equation}
f_{\#}(T)\left(\omega\right)=T\left(\gamma\wedge f^{\#}(\omega)\right),\quad\text{for all }\omega\in\D^{r}(V),\label{eq:pushforward_current_smooth}
\end{equation}
where $\gamma\in\D^{0}(\oset)$ is any cutoff function satisfying
\[
\spt(T)\cap f^{-1}\left(\spt(\omega)\right)\subset\mathrm{Int}\left\{ x\mid\gamma(x)=1\right\} .
\]
 For $T\in\D_{r}(\oset)$ with $\spt(T)\subset\compact$, where $K$
is a compact subset of $\oset$, we have the following bounds 
\begin{equation}
\begin{split}\Fmass{f^{\#}(T)} & \leq\sup_{x\in\compact}\|Df(x)\|^{r}\Fmass T,\\
\Fnormal{f^{\#}(T)} & \leq\sup_{x\in\compact}\left\{ \|Df(x)\|^{r-1},\|Df(x)\|^{r}\right\} \Fnormal T,\\
F_{f\left\{ \compact\right\} }\left(f^{\#}(T)\right) & \leq\sup_{x\in\compact}\left\{ \|Df(x)\|^{r},\|Df(x)\|^{r+1}\right\} F_{\compact}(T).
\end{split}
\label{eq:f(T)_bounds_smooth}
\end{equation}

Let $\oset\subset\reals^{n}$ be an open set and let $f$ and $g$
be smooth maps of $\oset$ into $\reals^{m}$. For an open set $A$
of $\reals$ such that $\left[a,b\right]\subset A$, a smooth homotopy
between the maps $f$ and $g$ is a map 
\begin{equation}
h:A\times\oset\to\reals^{m},
\end{equation}
such that 
\begin{equation}
h\left(a,x\right)=f(x),\;\text{and }h(b,x)=g(x),\quad\text{for all }x\in\oset.
\end{equation}
Henceforth, the following notation will be used 
\begin{equation}
h_{\tau}(x)=h(\tau,x),\;\text{for all}\, x\in\oset,
\end{equation}
and

\[
\dot{h}_{\tau}:\oset\to\reals^{m},\qquad\dot{h}_{\tau}(x)=Dh(\tau,x)\left(1,0\right)=\frac{\partial h}{\partial\tau}(\tau,x),\;\text{for all }x\in\oset.
\]
 For $T\in\D_{r}(\oset)$ and a homotopy $h$ between $f$ and $g$,
the \textit{$h$-deformation chain of $T$} is defined as the current
\begin{equation}
h_{\#}\left(\left[a,b\right]\times T\right)\in\D_{r+1}\left(\reals^{m}\right).
\end{equation}
Traditionally, the interval $[a,b]$ is taken as the unit interval
$[0,1]$. The properties $h_{\#}\left(\left[a,b\right]\times T\right)$
are further investigated in \cite[Section 4.1.9]{Federer1969} and
\cite[sec.~2.3]{Giaquinta1998}. A fundamental tool is the following
formula 
\begin{equation}
g_{\#}\left(T\right)-f_{\#}\left(T\right)=\bnd h_{\#}\left(\left[a,b\right]\times T\right)+h_{\#}\left(\left[a,b\right]\times\bnd T\right),\label{eq:homotopy_formula_for_currents}
\end{equation}
which is referred to as \emph{the homotopy} \emph{formula for currents.}

Let $\Lmap:\oset\to V$ be a locally Lipschitz map, the image of a
general $r$-current under a locally Lipschitz map is generally undefined
as for any $\omega\in D^{r}\left(V\right)$ the fullback $\Lmap^{\#}\left(\omega\right)$
need not be a smooth differential $r$-form in $\oset$, moreover,
the coefficients of $\Lmap^{\#}(\omega)$ are not necessary Borel
functions. For a normal current $T\in N_{r}\left(\oset\right)$ one
can define, see \cite[Sec.~2.3]{Giaquinta1998}, the pushforward of
$T$ by the locally Lipschitz map $\Lmap$ as the following weak limit
\[
\Lmap_{\#}\left(T\right)(\omega)=\lim_{\rho\to0}\left\{ \left(\left(\mol{\rho}\Lmap\right)_{\#}T\right)(\omega)\right\} ,
\]
where $\left\{ \mol{\rho}\Lmap\right\} _{\rho}$ is a sequence of
smooth approximations obtained by mollification of $\Lmap$ (see \cite[Section 4.1.2]{Federer1969}).
The strong convergence of the sequence is proven by Equation (\ref{eq:homotopy_formula_for_currents}),
as the sequence $\left\{ \left(\mol{\rho}\Lmap\right)_{\#}T\right\} _{\rho}$
is shown to be a Cauchy sequence in the Banach space of flat $r$-chains.
(See \cite[Section 4.1.14]{Federer1969}.) 

The operator $\Lmap_{\#}:N_{r,\compact}\left(\oset\right)\to N_{r,\Lmap\left\{ \compact\right\} }\left(V\right)$
is continuous with respect to the flat norm and thus extends (we keep
the same notation) to $\Lmap_{\#}:F_{r,\compact}\left(\oset\right)\to F_{r,\Lmap\left\{ \compact\right\} }\left(V\right)$.
For $T\in\D_{r}\left(\oset\right)$ with $\spt\left(T\right)\subset\compact$,
the bounds presented in Equation (\ref{eq:f(T)_bounds_smooth}) are
replaced with 
\begin{equation}
\begin{split}\Fmass{\Lmap_{\#}(T)} & \leq\sup_{x\in\compact}\left(\Lip_{\Lmap,\compact}\right)^{r}\Fmass T,\\
\Fnormal{\Lmap_{\#}(T)} & \leq\sup_{x\in\compact}\left\{ \left(\Lip_{\Lmap,\compact}\right)^{r-1},\left(\Lip_{\Lmap,\compact}\right)^{r}\right\} \Fnormal T,\\
F_{f\left\{ \compact\right\} }\left(\Lmap_{\#}(T)\right) & \leq\sup_{x\in\compact}\left\{ \left(\Lip_{\Lmap,\compact}\right)^{r},\left(\Lip_{\Lmap,\compact}\right)^{r+1}\right\} F_{\compact}(T).
\end{split}
\label{eq:f(T)_bounds_Lipschitz}
\end{equation}
For $T\in N_{r,\compact}\left(\oset\right)$ the existence of $\Lmap_{\#}(T)\in F_{r,\Lmap\left\{ \compact\right\} }\left(V\right)$,
and the second bound in Equation (\ref{eq:f(T)_bounds_Lipschitz}),
imply that $\Lmap_{\#}(T)\in N_{r,\Lmap\left\{ \compact\right\} }\left(V\right)$. 

Alternatively, one may define the pushforward $\Lmap_{\#}(T)$ by
utilizing the duality of flat chains and flat forms and setting 
\[
\Lmap_{\#}(T)(\omega)=\cochain_{\Lmap^{\#}(\omega)}(T),\quad\text{for all }\omega\in\D^{r}\left(\omega\right).
\]
 By Rademacher\textquoteright s theorem the derivative a Lipschitz
mapping exists for $\lusb^{n}$-almost every $x\in\oset$. Thus, Equation
(\ref{eq:pullback_form_smooth}) is meaningful for $\lusb^{n}$-almost
every $x\in\oset$ and $\Lmap^{\#}(\omega)$ is a flat $r$-form in
$\oset$. It follows that $\cochain_{\Lmap^{\#}(\omega)}$ is a flat
$r$-cochain and the action $\cochain_{\Lmap^{\#}(\omega)}(T)$ is
well defined. The homotopy theorem for currents, and in particular,
the homotopy formula given in Equation (\ref{eq:homotopy_formula_for_currents})
discussed above for smooth maps, is therefore extended to maps $h:A\times\oset\to\reals^{m}$
which are locally Lipschitz maps. We note that a similar definition
and Wolfe's representation theorem are applied in \cite[Sction X.9]{Whitney1957},
to define the pullback of a flat form by a Lipschitz map.

\section{The Lie derivative \label{sec:Lie-derivative}}

In this section we examine the regularity of the Lie derivative of
a differential $r$-form. Cartan's (magic) formula is a key element
in the following analysis and as a first step we examine the contraction
of a differential form by a smooth vector field.

We first introduce a component representation that will be useful
throughout this section. The summation convention will be used unless
otherwise stated. Let $v=v^{i}e_{i}$, $e_{i}=\nicefrac{\partial}{\partial x^{i}}$,
and $\omega=\omega_{\lambda}dx^{\lambda}$ with $\lambda\in\Lambda\left(n,r+1\right)$.
Then,
\begin{equation}
\omega\rest v=\left(\omega_{\lambda}dx^{\lambda}\right)\rest\left(v^{i}e_{i}\right)=v^{i}\omega_{\lambda}dx^{\lambda}\rest e_{i},\label{eq:v_omega_com}
\end{equation}
 and the exterior derivative of $\omega\rest v$ is given by 
\begin{equation}
\begin{split}\cbnd(\omega\rest v) & =\left(v^{i}\omega_{\lambda}\right)_{,j}dx^{j}\wedge\left(dx^{\lambda}\rest e_{i}\right),\\
 & =\left(v_{,j}^{i}\omega_{\lambda}+v^{i}\omega_{\lambda,j}\right)dx^{j}\wedge\left(dx^{\lambda}\rest e_{i}\right).
\end{split}
\label{eq:d(v_omega)_com}
\end{equation}
Consider the $M_{\compact}$-seminorm of $\omega\rest v$ and recall
that 
\begin{equation}
\Flmass{\compact}{\omega\rest v}=\ess\sup_{x\in\compact}\left\{ \|\left(\omega\rest v\right)(x)\|\right\} .\label{eq:M_K(v_w)}
\end{equation}
The $r$-form $\omega\rest v$ has $C(n,r)$ components each of which
is a sum of $(n-r)$ terms each of which is a multiplication of a
component of $\omega$ with a component of $v$. Hence, 
\begin{equation}
\Flmass{\compact}{\omega\rest v}\leq C(n,r)\sup_{x\in\compact}\|v(x)\|\Flmass{\compact}{\omega}.\label{eq:Mass_v_w}
\end{equation}

\begin{rem}
\label{rem:v_wedge_T_not_flat}For a smooth vector field, $v:\oset\to\reals^{n}$,
and $\omega\in\D^{r}\left(\oset\right)$ such that $\spt(\omega)\subset\compact$,
note that $\spt\left(\cbnd(\omega\rest v)\right)\subset\compact$
and the components of $\cbnd(v\irest\omega)$ are functions in $C_{0}^{\infty}\left(\oset\right)$.
Thus, 
\[
\Fmass{\cbnd(\omega\rest v)}=\ess\sup_{x\in\compact}\left\{ \|\cbnd(\omega\rest v)(x)\|\right\} <\infty.
\]
However, one cannot find a $C<\infty$ such that 
\[
\sup_{x\in\compact}\left\{ \|\cbnd(\omega\rest v)(x)\|\right\} \leq C\norm v{\Lip,\compact}\sup_{x\in\compact}\left\{ \|\cbnd\omega(x)\|\right\} .
\]
 \end{rem}
\begin{defn}
The Lie derivative of the differential form $\omega\in\D^{r}\left(\oset\right)$
with respect to the vector field $v$ on $\oset$ is the differential
$r$-form in $\oset$ denoted by $\Lie_{v}\omega$ and defined by
\begin{equation}
\emb_{t}^{\#}\left(\Lie_{v}\omega\right)=\frac{d}{dt}\left(\emb_{t}^{\#}\omega\right),\label{eq:Lie_form_def}
\end{equation}
where $\emb:\reals\times\oset\to\oset$ is the flow associated with
the vector field $v$ \cite[p.~370]{Marsden1988}. \label{def:Lie-derivative}
\end{defn}
The classical Cauchy-Lipschitz theory asserts the existence of a
flow for $v$, a time dependent vector field which is Lipschitz continuous
in the spatial variable and uniformly continuous with respect to the
time variable \cite[Sec. IV.1]{Lang1999}. The existence of flows
for vector fields with reduced regularity, such as vector fields of
bounded variation, is an active field of research and for further
discussion see \cite{Bouchut2006} and references cited therein. It
is noted that in the rest of this work the existence of a flow for
the Lipschitz vector field follows from the assumptions regarding
the motion described below. 

The Lie derivative of a differential form $\omega$ with respect to
the vector field $v$ satisfies the identity 
\begin{equation}
\Lie_{v}\omega=\cbnd\left(\omega\rest v\right)+\cbnd\omega\rest v,\label{eq:Cartan_forms}
\end{equation}
which is commonly known as \textit{Cartan's (magic) formula}. Note
that in case $v$ is a smooth vector field on $\oset$, it follows
that $\Lie_{v}\omega\in\D^{r}\left(\oset\right)$ for every $\omega\in\D^{r}\left(\oset\right)$.
\begin{lem}
Let $\omega\in\D^{r}\left(\oset\right)$ and let $v:\oset\to\reals^{n}$
be a smooth vector field. Then, there exists a constant $C\left(n,r\right)$
such that 
\begin{equation}
\Flmass{\compact}{\Lie_{v}\omega}\leq C(n,r)S_{\compact}(\omega)\norm v{\Lip,\compact}.\label{eq:M_Lie_deriv}
\end{equation}
Moreover, the Lie derivative of a sharp $r$-form with respect to
a smooth vector field, taken in the weak sense, is an $r$-form of
locally finite mass.\label{lem:Reugularity_of_Lie_derivative}\end{lem}
\begin{proof}
By Equations (\ref{eq:v_omega_com}), (\ref{eq:d(v_omega)_com}) and
(\ref{eq:Cartan_forms}), a local representation of the Lie derivative
$\Lie_{v}\omega$ is given by 
\begin{equation}
\Lie_{v}\omega=D(\omega)(v)+v_{,j}^{i}\omega_{\lambda}dx^{j}\wedge\left(dx^{\lambda}\rest e_{i}\right).\label{eq:Lie_derive_components}
\end{equation}
For a given $\lambda\in\bigwedge(n,r)$, the contraction $dx^{\lambda}\rest e_{i}$
does not vanish for $r$ base vectors $e_{k}$.For a selection of
$\lambda$ and $k$ such that $dx^{\lambda}\rest e_{k}\not=0$, the
wedge product $dx^{j}\wedge\left(dx^{\lambda}\rest e_{k}\right)$
will not vanish for a subset of $\{dx^{j}\}$ containing $(n-(r-1))$
elements. Hence,
\begin{equation}
\begin{split}\Flmass{\compact}{\Lie_{v}\omega} & \leq\Flmass{\compact}{D(\omega)(v)}+\Flmass{\compact}{v_{,j}^{i}\omega_{\lambda}dx^{j}\wedge\left(dx^{\lambda}\rest e_{i}\right)},\\
 & \leq\Lip_{\omega,\compact}\norm v{\infty,\compact}+\frac{n!}{r!\left(n-r\right)!}r(n-(r-1))\Flmass{\compact}{\omega}\Lip_{v,\compact}\\
 & \leq C(n,r)\norm v{\Lip,\compact}S_{\compact}(\omega).
\end{split}
,\label{eq:M(Lie(w))<S(w)}
\end{equation}
The extension to sharp forms follows from Rademacher's theorem which
implies that for a sharp form $\omega$, $D\omega$ exists almost
everywhere in $\oset$.
\end{proof}
For a smooth vector field $v:\oset\to\reals^{n}$ defined on the open
set $\oset\subset\reals^{n}$ and an $r$-current $T\in\D_{r}\left(\oset\right)$,
the $\left(r+1\right)$-current $v\wedge T$ is defined by 
\begin{equation}
v\wedge T\left(\omega\right)=T\left(\omega\rest v\right),\quad\text{for all }\omega\in\D^{r+1}\left(\oset\right).\label{eq:v_wedge_smooth}
\end{equation}
 As $v$ is a smooth vector field it follows that $\omega\rest v\in\D^{r}\left(\oset\right)$
and $T\left(\omega\rest v\right)$ is well defined. 

By the remark preceding Definition \ref{def:Lie-derivative}, it follows
that given $T\in F_{r,\compact}\left(\oset\right)$ and a smooth vector
field $v:\oset\to\reals^{n}$, the $(r+1)$-current $v\wedge T$ is
not necessarily a flat $(r+1)$-chain. Moreover, even if $T\in N_{r,\compact}\left(\oset\right)$
the current $v\wedge T$ may not be a flat $(r+1)$-chain. The analysis
of $v\wedge T$ is a key element in what follows and the foregoing
remark is an example for the restricted applicability of the flat
norm.

The contraction of a vector field and a differential form defined
above may be extended to include locally Lipschitz vector fields.
Let $\omega\in\D^{r+1}\left(\oset\right)$, where $\oset\subset\reals^{n}$
is an open set, and let $v:\oset\to\reals^{n}$ be a locally Lipschitz
vector field. Define $\omega\rest v$ as the pointwise limit of the
contractions with the mollified vector fields $\mol{\rho}v$, \emph{i.e.,}
\[
\omega\rest v(x)=\lim_{\rho\to0}\left\{ \omega(x)\rest\left(\mol{\rho}v\right)(x)\right\} .
\]
 We note that the convergence of the above limit is locally uniform
with respect to $x$. In a similar manner to the estimate in Equation
(\ref{eq:Mass_v_w}) we have 
\[
\begin{split}\Flmass{\compact}{\omega\rest v} & =\lim_{\rho\to0}\Flmass{\compact}{\omega\rest\left(\mol{\rho}v\right)}\\
 & \leq\lim_{\rho\to0}C(n,r)\ess\sup_{x\in\compact}\|\mol{\rho}v(x)\|\Flmass{\compact}{\omega}\\
 & =C\left(n,r\right)\|v\|_{\infty,\compact}\Flmass{\compact}{\omega}.
\end{split}
\]
For the Lie derivative of $\omega\in\D^{r}\left(\oset\right)$ with
respect to the Lipschitz vector field we have 
\[
\Lie_{v}\omega=\lim_{\rho\to0}\left(\Lie_{\mol{\rho}v}\omega\right)=\lim_{\rho\to0}\left(\cbnd\left(\omega\rest\left(\mol{\rho}v\right)\right)+\cbnd\omega\rest\left(\mol{\rho}v\right)\right),
\]
where the above limit is taken with respect to the $\Flmass{\compact}{\cdot}$-semi-norm.
The existence of the limit follows from the fact that 
\[
\lim_{\rho\to0}\norm{\mol{\rho}v-v}{\Lip,\compact}=0,
\]
and the bound given in Equation (\ref{eq:M(Lie(w))<S(w)}). Moreover,
the bound in Equation (\ref{eq:M(Lie(w))<S(w)}) holds for locally
Lipschitz vector field.

\section{Smooth Configurations and Motions \label{sec:Smooth_motion}}

Let $\body\subset\reals^{n}$ and $\ti\subset\reals$ bounded open
subsets. Recalling (\ref{eq:DefineEmb}), a \emph{smooth motion} $\motion$
defined over the time interval $\ti$ is viewed as a curve 
\begin{equation}
\motion:\ti\to\Emb{\body,\reals^{n}}.\label{eq:Def:motion}
\end{equation}
We assume that the motion $\motion$ is a $C^{1}$-curve with respect
to strong Lipschitz topology, that is, the derivative of the curve,
denoted by $\dot{\motion}$, is viewed as a curve 
\[
\dot{\motion}:\ti\to C^{\infty}\left(\body,\reals^{n}\right)\cap\Lip\left(\body,\reals^{n}\right),
\]
which is continuous with respect to the strong Lipschitz topology.

The motion $\motion$, induces a map 
\begin{equation}
\emb:\ti\times\body\to\reals^{n},\label{eq:emb_motion}
\end{equation}
by 
\[
\emb(\tau,x)=\motion(\tau)(x),\quad\text{for all}\;\tau\in\ti,\: x\in\body,
\]
and so 
\[
\frac{\partial\emb}{\partial t}(\tau,x)=\dot{\motion}(\tau)(x),\quad\text{for all}\;\tau\in\ti,\: x\in\body.
\]
It follows from Lemma \ref{lem:compact_strong_topology}, that there
exists a compact subset $\compm\subset\body$ such that for any $x\not\in\compm$
and every $t,\, t'\in\ti$
\begin{equation}
\emb(t,x)=\emb\left(t',x\right).\label{eq:emb_on_compact}
\end{equation}
Hence, for $x\not\in\compm$ and $t\in\ti$ 
\begin{equation}
\dot{\emb}\left(t,x\right)=0.\label{eq:velocity_on_compact}
\end{equation}
For some $t\in\ti$, set $\body'=\emb_{t}\left\{ \body\right\} $
and $\compm'=\emb_{t}\left\{ \compm\right\} $. By the preceding argument,
$B'$ and $\compm'$ are independent of the particular choice of $t\in\ti$. 
\begin{rem}
Equations (\ref{eq:emb_on_compact}, \ref{eq:velocity_on_compact})
and the existence of $\compact_{\motion}$ are key features of the
motion examined and stem from the use of the strong Lipschitz topology.
A drawback to the use of the strong Lipschitz topology is in the relatively
small supply of converging sequences of maps in the form of Equation
(\ref{eq:emb_motion}) converging to a motion as given in Equation
(\ref{eq:Def:motion}).
\end{rem}
For each $t\in\ti$, $\emb_{t}=\motion(t)\in\Emb{\body,\reals^{n}}$,
so there exists an inverse $\eta_{t}:\imag{\emb_{t}}=\body'\to\body$
such that $\emb_{t}\circ\eta_{t}=I_{\body'}$ with $I_{\body'}$ the
identity map on the set $\body'$. Consider the vector field 
\begin{equation}
v_{t}:\body'\to\reals^{n},\qquad v_{t}=\dot{\emb}_{t}\circ\eta_{t},\label{eq:EulerianVelocity-2}
\end{equation}
viewed as a vector field on $\body'$ such that $v_{t}(y)=0$ for
every $x\in\body'\backslash\compm'$. The vector field $v_{t}$ is
naturally extended to a vector field $\ve_{t}:\reals^{n}\to\reals^{n}$,
by setting 
\begin{equation}
\ve_{t}(x)=\begin{cases}
v_{t}(x),\: & x\in\body',\\
0, & x\not\in\body'.
\end{cases}\label{eq:V_extension}
\end{equation}
Thus,$\ve_{t}$ is a smooth vector field which vanishes on $\reals^{n}\backslash\compm'$
. It follows that, 
\[
\ve:\ti\times\reals^{n}\to\reals^{n},
\]
 is a time dependent Lipschitz vector field defined on $\reals^{n}$.

For $s,\, t\in\ti$, define 
\begin{equation}
J_{s,t}(x)=\begin{cases}
\emb_{s}\circ\eta_{t}(x), & x\in\body',\\
x, & x\not\in\body'.
\end{cases}\label{eq:smooth_flow}
\end{equation}
As shown in \cite{Falach2014}, 
\begin{equation}
\frac{\partial J_{s,t}}{\partial s}(x)=\ve_{s}\left(J_{s,t}\left(x\right)\right),\quad J_{t,t}(x)=x.\label{eq:ODE_flow_smooth}
\end{equation}
The map $J_{s,t}$, is the flow associated with the time dependent
vector field $\ve$. For $\omega\in\D^{r}\left(\reals^{n}\right)$,
\[
\frac{\partial\left(J_{\tau,t}^{\#}\omega\right)}{\partial\tau}\mid_{\tau=s}=J_{s,t}^{\#}\left(\Lie_{\ve_{s}}\omega\right),
\]
and as $\emb_{\tau}^{\#}\omega=\emb_{t}^{\#}\left(J_{\tau,t}^{\#}\omega\right)$,
it follows from Definition \ref{eq:Lie_form_def} and a direct computation
that 
\begin{equation}
\frac{\partial\left(\emb_{\tau}^{\#}\omega\right)}{\partial\tau}\mid_{\tau=t}=\emb_{t}^{\#}\left(\Lie_{\ve_{t}}\omega\right).\label{eq:d/dt_pullback_smooth}
\end{equation}

We now derive a representation formula for the $\emb$-deformation
chain associated with the motion $\emb:\ti\times\body\to\reals$ and
a general current $T\in\D_{r}\left(\body\right)$. Let $[a,b]\subset\ti$,
then, applying Equation (\ref{eq:product_comput}) to $\omega\in\D^{r+1}\left(\reals^{n}\right)$,
one has 
\begin{equation}
\begin{split}\emb_{\#}\left(\left[a,b\right]\times T\right)\left(\omega\right) & =\left(\left[a,b\right]\times T\right)\left(\emb^{\#}\left(\omega\right)\right),\\
 & =\int_{a}^{b}T\left(\emb_{\tau}^{\#}\left(\omega\right)_{H}\right)d\lusb_{\tau}^{1},\\
 & =\int_{a}^{b}T\left(\emb_{\tau}^{\#}\left(\omega\right)\rest e_{t}\right)d\lusb_{\tau}^{1}.
\end{split}
\label{eq:integral_homotopy}
\end{equation}
In order to examine the $r$-form $\emb_{\tau}^{\#}\left(\omega\right)\rest e_{t}$,
we apply it to an $r$-vector $\xi$ which we can assume to be ``space-like'',
that is, $\xi=v_{1}\wedge\dots\wedge v_{r}$ with $v_{i}\in\reals^{n}$,
for $i=1,\dots,r$. Otherwise, $\left(\emb_{\tau}^{\#}\left(\omega\right)\rest e_{t}\right)(\xi)=0$,
identically. One obtains, 
\begin{equation}
\begin{split}\left(\emb_{\tau}^{\#}\left(\omega\right)(x)\rest e_{t}\right)(\xi) & =\emb^{\sharp}(\omega)(\tau,x)(e_{t}\wedge v_{1}\wedge\dots\wedge v_{r}),\\
 & =\omega\circ\emb_{\tau}(x)\left(D\emb(\tau,x)\left(e_{t}\right)\wedge\bigwedge_{r}D\emb_{\tau}(x)\left(\xi\right)\right),\\
 & =\omega\circ\emb_{\tau}(x)\left(\dot{\emb}_{\tau}(x)\wedge\bigwedge_{r}D\emb_{\tau}(x)\left(\xi\right)\right),\\
 & =((\omega\circ\emb_{\tau})\rest\dot{\emb}_{\tau})(x)\left(\bigwedge_{r}D\emb_{\tau}(x)(\xi)\right).
\end{split}
\end{equation}

As $\dot{\emb}_{\tau}(x)=v_{\tau}(\emb_{\tau}(x))=\ve_{\tau}\circ\emb_{\tau}(x),$
\begin{equation}
\begin{split}\left(\emb^{\#}(\omega)(\tau,x)\rest e_{t}\right)(\xi) & =\left((\omega\circ\emb_{\tau})\rest\left(\ve_{\tau}\circ\emb_{\tau}\right)\right)(x)\left(\bigwedge_{r}D\emb_{\tau}(x)(\xi)\right),\\
 & =(\omega\rest\ve_{\tau})\circ\emb_{\tau}(x)\left(\bigwedge_{r}D\emb_{\tau}(x)(\xi)\right),\\
 & =\emb_{\tau}^{\#}(\omega\rest\ve_{\tau})(x)(\xi).
\end{split}
\end{equation}
It is concluded that, 
\begin{equation}
\emb_{\tau}^{\#}(\omega)\rest e_{t}=\emb_{\tau}^{\#}(\omega\rest\ve_{\tau}).\label{eq:e_contraction_pullback_form}
\end{equation}

Returning to Equation (\ref{eq:integral_homotopy}), note that the
integrand may be rewritten as 
\begin{equation}
\begin{split}T(\emb_{\tau}^{\#}(\omega)\rest e_{t}) & =T(\emb_{\tau}^{\#}(\omega\rest\ve_{\tau})),\\
 & =\emb_{\tau\#}(T)(\omega\rest\ve_{\tau}),\\
 & =\ve_{\tau}\we\emb_{\tau\#}(T)\left(\omega\right),
\end{split}
\end{equation}
and Equation (\ref{eq:integral_homotopy}) assumes the form
\begin{equation}
\emb_{\#}\left(\left[a,b\right]\times T\right)\left(\omega\right)=\int_{a}^{b}(\ve_{\tau}\we\emb_{\tau\#}T)(\omega)d\lusb_{\tau}^{1}.\label{eq:emb_pushforward_T}
\end{equation}
Applying (\ref{eq:homotopy_formula_for_currents}) to (\ref{eq:homotopy_formula_for_currents}),
one finally has, 
\[
\begin{split}\left(\emb_{b\#}\left(T\right)-\emb_{a\#}\left(T\right)\right)\omega & =\left(\bnd\emb_{\#}\left(\left[a,b\right]\times T\right)+\emb_{\#}\left(\left[a,b\right]\times\bnd T\right)\right)\omega,\\
 & =\int_{a}^{b}\left[(\ve_{\tau}\we\emb_{\tau\#}T)(\cbnd\omega)+(\ve_{\tau}\we\emb_{\tau\#}\bnd T)(\omega)\right]d\lusb_{\tau}^{1},\\
 & =\int_{a}^{b}\left[(\emb_{\tau\#}T)\left(\cbnd\omega\rest\ve_{\tau}+\cbnd\left(\omega\rest\ve_{\tau}\right)\right)\right]d\lusb_{\tau}^{1},\\
 & =\int_{a}^{b}\left[(\emb_{\tau\#}T)\left(\Lie_{\ve_{\tau}}\omega\right)\right]d\lusb_{\tau}^{1}.
\end{split}
\]

\section{The Kinematics of Currents under a Smooth Motion \label{sec:kinematics_cerrents}}

This section is devoted to the examination of kinematic properties
of generalized domains. A generalized $r$-dimensional oriented domain,
is naturally viewed as an $r$-current. In the selection of the appropriate
class of domains, the collection of all $r$-currents is far greater
than what we would consider as suitable. The selection of the appropriate
class is motivated by the following guidelines. Firstly, a current
representing a generalized domain must have a local character, at
least in some measure theoretic sense. Secondly, such a current must
have a definite, quantitative notion of a boundary.Finally, such a
current should be well behaved under the image of a Lipschitz map.%
\footnote{We feel that these requirement are in the spirit put forth by Noll
\& Virga in \cite{Noll1988} where the class admissible bodies should
include all those that can be imagined by an engineer but exclude
those that can be dreamt up only by an ingenious mathematician.%
} The introductory discussion in Sections \ref{sec:Notation-and-Preliminaries}
and \ref{sec:fushforward_homotopy} indicates that a convenient choice
for the class of domains is the collection of flat chains of finite
mass. Thus, a prototypical control volume $T$, is viewed as a flat
$r$-chain of finite mass in $\body$. Using the properties of a motion
we outlined above and the corresponding notation of Section \ref{sec:Smooth_motion},
we consider $T\in F_{r,\compm}\left(\body\right)$, where $\compm\subset\body$
is the compact set containing the region where the motion is nontrivial.

\begin{rem}
Let $T\in F_{r}\left(\body\right)$ and $\gamma:\body\to\reals$ a
locally Lipschitz function, the multiplication $\gamma\wedge T$ is
flat $r$-chain. Thus, a flat $r$-chain may represent not only a
geometric domain but may also be represent some intensive property.
See \cite{Falach2013} for further details.
\end{rem}
Consider a map $\emb$ induced by a motion as defined in Equation
(\ref{eq:emb_motion}) and a flat chain $T\in F_{r,\compm}\left(\body\right)$
such that $\Fmass T<\infty$. The curve $t\mapsto\emb_{t\#}\left(T\right)$
will be viewed in this work as the time evolution of the control volume
represented by the current $T$. 
\begin{lem}
Let $\emb$ be the map associated with a motion as defined by Equation
(\ref{eq:emb_motion}) and $T\in F_{r,\compm}\left(\body\right)$
a flat chain of finite mass. The curve induced by the pushforward
$t\mapsto\emb_{t\#}T$ is a continuous curve with respect to the $M_{K_{m}^{'}}$-norm
on $\D_{r}\left(\reals^{n}\right)$.\label{lem:emb_t(T)_smooth_con}\end{lem}
\begin{proof}
Let $t\in\ti$ and select $\eps$ such that $t+\eps\in\ti$. Then,
\[
\begin{split}\Flmass{\compm'}{\emb_{t+\eps\#}T-\emb_{t\#}T} & =\sup_{\omega\in\D^{r}\left(\reals^{n}\right)}\frac{\left(\emb_{t+\eps\#}T-\emb_{t\#}T\right)\omega}{\Flmass{\compm'}{\omega}},\\
 & =\sup_{\omega\in\D^{r}\left(\reals^{n}\right)}\frac{T\left(\emb_{t+\eps}^{\#}\omega-\emb_{t}^{\#}\omega\right)}{\Flmass{\compm'}{\omega}},\\
 & \leq\sup_{\omega\in\D^{r}\left(\reals^{n}\right)}\frac{\Fmass T\Flmass{\compm}{\emb_{t+\eps}^{\#}\omega-\emb_{t}^{\#}\omega}}{\Flmass{\compm'}{\omega}},
\end{split}
\]
where the last line follows from the integral representation of $T$.
For the term $\Flmass{\compm}{\emb_{t+\eps}^{\#}\omega-\emb_{t}^{\#}\omega}$,
a direct computation shows that 
\[
\begin{split} & \Flmass{\compm}{\emb_{t+\eps}^{\#}\omega-\emb_{t}^{\#}\omega}\\
 & =\sup_{x\in\compm}\left\{ \sup_{\xi}\left\{ \omega\left(\emb_{t+\eps}\left(x\right)\right)\left[\bigwedge_{r}D\emb_{t+\eps}(x)(\xi)\right]-\omega\left(\emb_{t}\left(x\right)\right)\left[\bigwedge_{r}D\emb_{t}(x)(\xi)\right]\right\} \right\} ,\\
 & \leq\sup_{x\in\compm}\left\{ \sup_{\xi}\left\{ \omega\left(\emb_{t+\eps}\left(x\right)\right)\left[\bigwedge_{r}D\emb_{t+\eps}(x)(\xi)\right]-\omega\left(\emb_{t+\eps}\left(x\right)\right)\left[\bigwedge_{r}D\emb_{t}(x)(\xi)\right]\right\} \right\} \\
 & \quad+\sup_{x\in\compm}\left\{ \sup_{\xi}\left\{ \omega\left(\emb_{t+\eps}\left(x\right)\right)\left[\bigwedge_{r}D\emb_{t}(x)(\xi)\right]-\omega\left(\emb_{t}\left(x\right)\right)\left[\bigwedge_{r}D\emb_{t}(x)(\xi)\right]\right\} \right\} ,\\
 & \leq\sup_{x\in\compm}\left\{ \sup_{\xi}\left\{ \omega\left(\emb_{t+\eps}\left(x\right)\right)\left[\bigwedge_{r}D\left(\emb_{t+\eps}-\emb_{t}\right)(x)(\xi)\right]\right\} \right\} \\
 & \quad+\sup_{x\in\compm}\left\{ \sup_{\xi}\left\{ \left(\omega\left(\emb_{t+\eps}\left(x\right)\right)-\omega\left(\emb_{t}\left(x\right)\right)\right)\left[\bigwedge_{r}D\emb_{t}(x)(\xi)\right]\right\} \right\} ,\\
 & \leq\Flmass{\compm'}{\omega}\left(\Lip_{\emb_{t+\eps}-\emb_{t},\compm}\right)^{r}+\left(\Lip_{\emb_{t},\compm}\right)^{r}\Flmass{\compm'}{\omega\circ\emb_{t+\eps}-\omega\circ\emb_{t}}.
\end{split}
\]
By the continuity of the motion with respect to the strong Lipschitz
topology it follows that as $\eps\to0$ we have $\left(\Lip_{\emb_{t+\eps}-\emb_{t},\compm}\right)^{r}\to0$.
For the second term, as $\norm{\emb_{t+\eps}-\emb_{t}}{\infty,\compm}\to0$,
and since $\omega$ is smooth, it follows that $\Flmass{\compm'}{\omega\circ\emb_{t+\eps}-\omega\circ\emb_{t}}\to0$,
which completes the proof.
\end{proof}
As considered in \cite{Falach2014}, for all $\omega\in\D^{r}\left(\oset\right)$,
\[
\begin{split}\frac{d}{dt}\left(\left(\emb_{t\#}T\right)(\omega)\right)\mid_{t=\tau} & =\frac{d}{dt}\left(T\emb_{t}^{\#}(\omega)\right)\mid_{t=\tau},\\
 & =T\left(\frac{d}{dt}\emb_{t}^{\#}(\omega)\mid_{t=\tau}\right),\\
 & =T\left(\emb_{\tau}^{\#}\left(\Lie_{\ve_{\tau}}\omega\right)\right).
\end{split}
\]
Thus, using Equation (\ref{eq:Cartan_forms}) 
\[
\begin{split}T\left(\emb_{\tau}^{\#}\left(\Lie_{\ve_{\tau}}\omega\right)\right) & =\emb_{\tau\#}T\left(\cbnd\left(\omega\rest\ve_{\tau}\right)+\left(\cbnd\omega\right)\rest\ve_{\tau}\right),\\
 & =\left(\ve_{\tau}\wedge\bnd\left(\emb_{\tau\#}T\right)+\bnd\left(\ve_{\tau}\wedge\emb_{\tau\#}T\right)\right)\omega.
\end{split}
\]
The foregoing result applies to general currents and is not restricted
to flat chains of finite mass. It may also be written with the introduction
of $\Reyo_{\ve_{t}}=\left(\Lie_{\ve_{t}}\right)^{*}$ as the dual
operator of the Lie derivative, that is 
\[
T\left(\emb_{t}^{\#}\left(\Lie_{\ve_{t}}\omega\right)\right)=\Reyo_{\ve_{t}}\left(\emb_{t\#}T\right)\omega.
\]
Thus, for any de Rham current $T\in\D_{r}\left(\reals^{n}\right)$
\begin{equation}
\Reyo_{\ve_{t}}(T)=\ve_{t}\wedge\bnd T+\bnd\left(\ve_{t}\wedge T\right).\label{eq:Reynold_operator}
\end{equation}
It is observed that similar results have been reported in \cite{Chi2012}.

The previous analysis shows that the derivative $\nicefrac{d(\emb_{t\#}T)}{dt}$
converges in the topology of $\D_{r}\left(\reals^{n}\right)$. The
following theorem considers the convergence of the limit above in
the sharp norm topology.
\begin{thm}
\label{thm:driv_smooth_motion} Let $\emb$ be a map associated with
a motion as defined by Equation (\ref{eq:emb_motion}) and let $T\in F_{r,\compm}\left(\body\right)$
be a flat chain of finite mass. The derivative $\frac{d}{dt}\left(\emb_{t\#}T\right)\mid_{t=\tau}$,
exists in the topology of $\Fsharp{\compm'}{\reals^{n}}$ and is given
by
\begin{equation}
\frac{d}{dt}\left(\emb_{t\#}T\right)\mid_{t=\tau}=\bnd\left(\ve_{\tau}\wedge\emb_{\tau\#}T\right)+\ve_{\tau}\wedge\bnd\left(\emb_{\tau\#}T\right).\label{eq:driv_formal}
\end{equation}
 \end{thm}
\begin{proof}
By the homotopy formula (\ref{eq:homotopy_formula_for_currents}),
we may write 
\begin{equation}
\frac{d}{dt}\left(\emb_{t\#}T\right)=\lim_{\eps\to0}\frac{\bnd\emb_{\#}\left(\left[\tau,\tau+\eps\right]\times T\right)+\emb_{\#}\left(\left[\tau,\tau+\eps\right]\times\bnd T\right)}{\eps}.\label{eq:driv_by_homotopy}
\end{equation}
Thus,
\begin{multline*}
\frac{1}{\eps}\Fsharp{\compm'}{\emb_{t+\eps\#}T-\emb_{t\#}T-\eps\left(\bnd\left(\ve_{t}\wedge\emb_{t\#}T\right)+\ve_{t}\wedge\bnd\left(\emb_{t\#}T\right)\right)}\\
\begin{split} & =\frac{1}{\eps}\sup_{\omega\in\D^{r}\left(\reals^{n}\right)}\frac{\left(\emb_{t+\eps\#}T-\emb_{t\#}T-\eps\left(\bnd\left(\ve_{t}\wedge\emb_{t\#}T\right)+\ve_{t}\wedge\bnd\left(\emb_{t\#}T\right)\right)\right)\omega}{\Fsharp{\compm'}{\omega}},\\
 & =\frac{1}{\eps}\sup_{\omega\in\D^{r}\left(\reals^{n}\right)}\left\{ \frac{\int_{t}^{t+\eps}\left[(\ve_{\tau}\we\emb_{\tau\#}T)(\cbnd\omega)+(\ve_{\tau}\we\emb_{\tau\#}\bnd T)(\omega)\right]d\lusb_{\tau}^{1}}{\Fsharp{\compm'}{\omega}}\right.\\
 & \:\quad\qquad\qquad\left.-\frac{\eps\left(\ve_{t}\wedge\emb_{t\#}T\right)\left(\cbnd\omega\right)+\ve_{t}\wedge\bnd\left(\emb_{t\#}T\right)\omega}{\Fsharp{\compm'}{\omega}}\right\} ,\\
 & \leq\frac{1}{\eps}\sup_{\omega\in\D^{r}\left(\reals^{n}\right)}\left\{ \frac{\int_{t}^{t+\eps}\left[T\left(\emb_{\tau}^{\#}\Lie_{\ve_{\tau}}\omega\right)\right]d\lusb_{\tau}^{1}-\eps T\left(\left(\emb_{t}^{\#}\Lie_{\ve_{t}}\omega\right)\right)}{\Fsharp{\compm'}{\omega}}\right\} ,\\
 & \leq\sup_{\omega\in\D^{r}\left(\reals^{n}\right)}\left\{ \Fmass T\sup_{s\in[t,t+\eps]}\frac{\Flmass{\compm}{\emb_{s}^{\#}\Lie_{\ve_{s}}\omega-\emb_{t}^{\#}\Lie_{\ve_{t}}\omega}}{\Fsharp{\compm'}{\omega}}\right\} .
\end{split}
\end{multline*}
For the term $\Flmass{\compm}{\emb_{s}^{\#}\Lie_{\ve_{s}}\omega-\emb_{t}^{\#}\Lie_{\ve_{t}}\omega}$,
one has the estimate
\begin{multline*}
\Flmass{\compm}{\emb_{s}^{\#}\Lie_{\ve_{s}}\omega-\emb_{t}^{\#}\Lie_{\ve_{t}}\omega}\\
\begin{split}\leq & \Flmass{\compm}{\emb_{s}^{\#}\Lie_{\ve_{s}}\omega-\emb_{s}^{\#}\Lie_{\ve_{t}}\omega+\emb_{s}^{\#}\Lie_{\ve_{t}}\omega-\emb_{t}^{\#}\Lie_{\ve_{t}}\omega}\\
\leq & \Flmass{\compm}{\emb_{s}^{\#}\Lie_{\ve_{s}}\omega-\emb_{s}^{\#}\Lie_{\ve_{t}}\omega}+\Flmass{\compm}{\emb_{s}^{\#}\Lie_{\ve_{t}}\omega-\emb_{t}^{\#}\Lie_{\ve_{t}}\omega}\\
\leq & \Flmass{\compm}{\emb_{s}^{\#}\left(\Lie_{\ve_{s}-\ve_{t}}\omega\right)}+\Flmass{\compm}{\emb_{s}^{\#}\Lie_{\ve_{t}}\omega-\emb_{t}^{\#}\Lie_{\ve_{t}}\omega}.
\end{split}
\end{multline*}
By Equation (\ref{eq:M_Lie_deriv}), 
\begin{eqnarray*}
\Flmass{\compm}{\emb_{s}^{\#}\left(\Lie_{\ve_{s}-\ve_{t}}\omega\right)} & \leq & \left(\Lip_{\emb_{s},\compm}\right)^{r}\Flmass{\compm'}{\Lie_{\ve_{s}-\ve_{t}}\omega},\\
 & \leq & \left(\Lip_{\emb_{s},\compm}\right)^{r}C(n,r)S_{\compm'}(\omega)\norm{\ve_{s}-\ve_{t}}{\Lip,\compm'},
\end{eqnarray*}
and since $\norm{\ve_{s}-\ve_{t}}{\Lip,\compm'}\to0$ as $s\to t$,
it follows that 
\[
\lim_{s\to t}\Flmass{\compm}{\emb_{s}^{\#}\left(\Lie_{\ve_{s}-\ve_{t}}\omega\right)}=0.
\]
Applying Lemma \ref{lem:emb_t(T)_smooth_con}, it follows that 
\[
\lim_{s\to t}\Flmass{\compm}{\emb_{s}^{\#}\Lie_{\ve_{t}}\omega-\emb_{t}^{\#}\Lie_{\ve_{t}}\omega}=0.
\]
Setting 
\[
T_{k}=\frac{\bnd\emb_{\#}\left(\left[t,t+\frac{1}{k}\right]\times T\right)+\emb_{\#}\left(\left[t,t+\frac{1}{k}\right]\times\bnd T\right)}{\frac{1}{k}},
\]
we obtain a sequence $\left\{ T_{k}\right\} $ of flat $r$-chains
converging to $\ve_{t}\wedge\bnd\left(\emb_{t\#}T\right)+\bnd\left(\ve_{t}\wedge\emb_{t\#}T\right)$
in $S_{r}\left(\reals^{n}\right)$. As normal currents are dense in
the space of flat chains one may obtain a sequence of normal $r$-currents
converging to $\ve_{t}\wedge\bnd\left(\emb_{t\#}T\right)+\bnd\left(\ve_{t}\wedge\emb_{t\#}T\right)$
in the $S_{\compm'}$-norm. \end{proof}
\begin{rem}
Note that the map $t\mapsto\bnd\left(\ve_{t}\wedge\emb_{t\#}T\right)+\hat{v}_{t}\wedge\bnd\left(\emb_{t\#}T\right)$
need not be continuous with respect to the $S_{K_{m}^{'}}$-norm,
as $\Fsharp{\compact'}{\omega}$ may depend of the second derivative
of $\omega\in\D^{r}\left(\reals^{n}\right)$. As an example, consider
the simple case of $\body\subset\reals$, and $T\in N(\oset)$ given
by $T(\omega)=\omega(x_{0})+\cbnd\omega(x_{0})$. For $v=e_{1}$,
\[
T\left(\Lie_{v}\omega\right)=D\omega(x_{0})+D^{2}\omega(x_{0}).
\]

\end{rem}

\section{Lipschitz Type Configurations and Motion \label{sec:Lipschitz_motion}}

This section extends the foregoing discussion to non smooth motions.
In particular, configurations represented by bi-Lipschitz maps, as
well as the corresponding motions, will be examined.

The definition of a motion, as introduced in Section \ref{sec:Smooth_motion},
is generalized by considering curves of the form 
\begin{equation}
\motion:\ti\to\Limb\left(\body,\reals^{n}\right),\label{eq:Def:motion_Lipschitz}
\end{equation}
which we assume are continuously differentiable with respect to the
strong Lipschitz topology. Thus, the time derivative of the map is
\[
\dot{\motion}:\ti\to\Lip\left(\body,\reals\right),
\]
a continuous curve with respect to the strong Lipschitz topology.
The motion $m$ induces a map 
\[
\conf:\ti\times\body\to\reals^{n},
\]
such that 
\begin{equation}
\conf(\tau,x)=\motion(\tau)(x),\quad\text{for all}\;\tau\in\ti,\: x\in\body,\label{eq:Lips_map_motion}
\end{equation}
and so 
\[
\dot{\conf}(\tau,x)=\dot{\motion}(\tau)(x),\quad\text{for all}\;\tau\in\ti,\: x\in\body.
\]

Using the results of Fukui \cite{Fukui2005}, we can make the analogous
definitions of $\compm$, $\compm'$, and $\body'$ as in Section
\ref{sec:Smooth_motion}. We consider the flow $J_{s,t}$ and the
vector field $\ve:\ti\times\reals^{n}\to\reals^{n}$ as in the smooth
case replacing $\emb$ with $\conf$. The vector field 
\[
\ve:\ti\times\reals^{n}\to\reals^{n},
\]
 is a time dependent Lipschitz vector field on $\reals^{n}$. As in
the smooth case, we consider the flow $J_{s,t}$ associated with the
time dependent vector field $\ve$ satisfying 
\[
\frac{\partial J_{s,t}}{\partial s}(x)=\ve_{s}\left(J_{s,t}\left(x\right)\right),\quad J_{t,t}(x)=x.
\]
For $\omega\in\D^{r}\left(\reals^{n}\right)$ it follows that 
\[
\frac{\partial(J_{\tau,t}^{\#}\omega)}{\partial\tau}\mid_{\tau=s}=J_{s,t}^{\#}\left(\Lie_{\ve_{s}}\omega\right),
\]
where $\Lie_{\ve_{s}}\omega$ is the Lie derivative of $\omega$ with
respect to Lipschitz vector field $\ve_{s}$ as discussed is Section
\ref{sec:Lie-derivative}. In particular, 
\begin{equation}
\frac{\partial(\conf_{\tau}^{\#}\omega)}{\partial\tau}\mid_{\tau=t}=\conf_{t}^{\#}\left(\Lie_{\ve_{t}}\omega\right).\label{eq:d/dt_pullback_Lipschitz}
\end{equation}

We note that for a given $\omega\in\D^{r}\left(\reals^{n}\right)$,
the curve $t\mapsto\conf_{t}^{\#}\omega$ is valued in the space of
flat $r$-forms in $\body$. As in section \ref{sec:Smooth_motion},
the curve $t\mapsto\conf_{t\#}\left(T\right)$ is used to model the
time evolution of the generalized domain. The main results described
in Section \ref{sec:kinematics_cerrents} apply in the case of Lipschitz
motions. 
\begin{lem}
\label{lem:conf_t(T)_Lipschitz}For the mapping $\conf$, as defined
in Equation (\ref{eq:Lips_map_motion}), and a flat chain $T\in F_{r,\compm}\left(\body\right)$
with $\Fmass T<\infty$, the curve $t\mapsto\conf_{t\#}\left(T\right)$,where
$\conf_{t\#}\left(T\right)$ is defined in Section \ref{sec:fushforward_homotopy}
by $\conf_{t\#}\left(T\right)=\lim_{\rho\to0}\left(\mol{\rho}\conf_{t}\right)\left(T\right)$,
is continuous with respect to the $M_{\compm'}$-norm on $\D_{r}\left(\reals^{n}\right)$. \end{lem}
\begin{proof}
let $t\in\ti$ and select $\eps$ such that $t+\eps\in\ti$. Then,
\begin{multline*}
\Flmass{\compm'}{\conf_{t+\eps\#}T-\conf_{t\#}T}\\
\begin{split} & =\sup_{\omega\in\D^{r}\left(\reals^{n}\right)}\left\{ \lim_{\rho\to0}\frac{T\left(\left(\mol{\rho}\conf_{t+\eps}\right)^{\#}\omega-\left(\mol{\rho}\conf_{t}\right)^{\#}\omega\right)}{\Flmass{\compm'}{\omega}}\right\} ,\\
 & \leq\sup_{\omega\in\D^{r}\left(\reals^{n}\right)}\left\{ \lim_{\rho\to0}\frac{\Fmass T\Flmass{\compm}{\left(\mol{\rho}\conf_{t+\eps}\right)^{\#}\omega-\left(\mol{\rho}\conf_{t}\right)^{\#}\omega}}{\Flmass{\compm'}{\omega}}\right\} ,
\end{split}
\end{multline*}
 and as in the proof of Lemma \ref{lem:emb_t(T)_smooth_con}
\begin{multline*}
\Flmass{\compm}{\left(\mol{\rho}\conf_{t+\eps}\right)^{\#}\omega-\left(\mol{\rho}\conf_{t}\right)^{\#}\omega}\\
\begin{split} & \leq\Flmass{\compm'}{\omega}\left(\Lip_{\left(\mol{\rho}\conf_{t+\eps}-\mol{\rho}\conf_{t}\right),\compm}\right)^{r}\\
 & \quad+\left(\Lip_{\mol{\rho}\conf_{t},\compm}\right)^{r}\Flmass{\compm'}{\omega\circ\left(\mol{\rho}\conf_{t+\eps}\right)-\omega\circ\left(\mol{\rho}\conf_{t}\right)}.
\end{split}
\end{multline*}
 In addition, 
\[
\lim_{\rho\to0}\Lip_{\left(\mol{\rho}\conf_{t+\eps}-\mol{\rho}\conf_{t}\right),\compm}=\lim_{\rho\to0}\Lip_{\left(\mol{\rho}\left(\conf_{t+\eps}-\conf_{t}\right)\right),\compm}=\Lip_{\conf_{t+\eps}-\conf_{t},\compm}.
\]
We now consider the limit of the foregoing estimate when $\eps\to0$.
By the continuity of the motion it follows that $\Lip_{\conf_{t+\eps}-\conf_{t},\compm}\to0$
as $\eps\to0$. The second term in the estimate above vanishes by
the continuity of $\omega$ and the continuity of the mollified motion
$\mol{\rho}\conf$ with respect to the strong Lipschitz topology.
 In conclusion, 
\[
\lim_{\eps\to0}\Flmass{\compm'}{\conf_{t+\eps\#}T-\conf_{t\#}T}=0.
\]

\end{proof}
The following theorem is a generalization of Theorem \ref{thm:driv_smooth_motion}.
\begin{thm}
Let $\conf$ be a map associated with a motion as defined by Equation
(\ref{eq:Lips_map_motion}) and let $T\in F_{r,\compm}\left(\body\right)$
be a flat chain of finite mass. The derivative $\frac{d}{dt}\left(\conf_{t\#}T\right)\mid_{t=\tau}$,
exists in the topology of $\Fsharp{\compm'}{\reals^{n}}$ and is given
by
\begin{equation}
\frac{d}{dt}\left(\conf_{t\#}T\right)\mid_{t=\tau}=\bnd\left(\ve_{\tau}\wedge\conf_{\tau\#}T\right)+\ve_{\tau}\wedge\bnd\left(\conf_{\tau\#}T\right).\label{eq:ddt(conf(T))}
\end{equation}
\end{thm}
\begin{proof}
We have to show that
\[
\lim_{\eps\to0}\Fsharp{\compm'}{\frac{\conf_{\tau+\eps\#}T-\conf_{\tau\#}T}{\eps}-\left(\bnd\left(\ve_{\tau}\wedge\conf_{\tau\#}T\right)+\ve_{\tau}\wedge\bnd\left(\conf_{\tau\#}T\right)\right)}=0.
\]
In an analogous manner to the proof of theorem \ref{thm:driv_smooth_motion},
and as the homotopy theorem holds for the Lipschitz case, it follows
that 

\[
\begin{split} & \frac{1}{\eps}\Fsharp{\compm'}{\emb_{t+\eps\#}T-\emb_{t\#}T-\eps\left(\bnd\left(\ve_{t}\wedge\emb_{t\#}T\right)+\ve_{t}\wedge\bnd\left(\emb_{t\#}T\right)\right)}\\
 & \quad\leq\frac{1}{\eps}\sup_{\omega\in\D^{r}\left(\reals^{n}\right)}\left\{ \frac{\int_{t}^{t+\eps}\left[T\left(\emb_{\tau}^{\#}\Lie_{\ve_{\tau}}\omega\right)\right]d\lusb_{\tau}^{1}-\eps T\left(\left(\emb_{t}^{\#}\Lie_{\ve_{t}}\omega\right)\right)}{\Fsharp{\compm'}{\omega}}\right\} ,\\
 & \quad\leq\sup_{\omega\in\D^{r}\left(\reals^{n}\right)}\left\{ \Fmass T\sup_{s\in[t,t+\eps]}\frac{\Flmass{\compm}{\emb_{s}^{\#}\Lie_{\ve_{s}}\omega-\emb_{t}^{\#}\Lie_{\ve_{t}}\omega}}{\Fsharp{\compm'}{\omega}}\right\} .
\end{split}
\]
The term $\Flmass{\compm}{\emb_{s}^{\#}\Lie_{\ve_{s}}\omega-\emb_{t}^{\#}\Lie_{\ve_{t}}\omega}$
may be estimated as
\begin{multline*}
\Flmass{\compm}{\emb_{s}^{\#}\Lie_{\ve_{s}}\omega-\emb_{t}^{\#}\Lie_{\ve_{t}}\omega}\\
\begin{split}\leq & \Flmass{\compm}{\emb_{s}^{\#}\Lie_{\ve_{s}}\omega-\emb_{s}^{\#}\Lie_{\ve_{t}}\omega+\emb_{s}^{\#}\Lie_{\ve_{t}}\omega-\emb_{t}^{\#}\Lie_{\ve_{t}}\omega},\\
\leq & \Flmass{\compm}{\emb_{s}^{\#}\Lie_{\ve_{s}}\omega-\emb_{s}^{\#}\Lie_{\ve_{t}}\omega}+\Flmass{\compm}{\emb_{s}^{\#}\Lie_{\ve_{t}}\omega-\emb_{t}^{\#}\Lie_{\ve_{t}}\omega},\\
\leq & \Flmass{\compm}{\emb_{s}^{\#}\left(\Lie_{\ve_{s}-\ve_{t}}\omega\right)}+\Flmass{\compm}{\emb_{s}^{\#}\Lie_{\ve_{t}}\omega-\emb_{t}^{\#}\Lie_{\ve_{t}}\omega}.
\end{split}
\end{multline*}
By Equation (\ref{eq:M_Lie_deriv}) we have 
\begin{eqnarray*}
\Flmass{\compm}{\emb_{s}^{\#}\left(\Lie_{\ve_{s}-\ve_{t}}\omega\right)} & \leq & \left(\Lip_{\emb_{s},\compm}\right)^{r}\Flmass{\compm'}{\Lie_{\ve_{s}-\ve_{t}}\omega},\\
 & \leq & \left(\Lip_{\emb_{s},\compm}\right)^{r}C(n,r)S_{\compm'}(\omega)\norm{\ve_{s}-\ve_{t}}{\Lip,\compm'}
\end{eqnarray*}
and since $\norm{\ve_{s}-\ve_{t}}{\Lip,\compm'}\to0$ as $s\to t$,
it follows that 
\[
\lim_{s\to t}\Flmass{\compm}{\emb_{s}^{\#}\left(\Lie_{\ve_{s}-\ve_{t}}\omega\right)}=0.
\]
Lemma \ref{lem:conf_t(T)_Lipschitz}, implies that 
\[
\lim_{s\to t}\Flmass{\compm}{\emb_{s}^{\#}\Lie_{\ve_{t}}\omega-\emb_{t}^{\#}\Lie_{\ve_{t}}\omega}=0
\]
which completes the proof. 
\end{proof}

\section{The Product Rule and the Transport Theorem \label{sec:Transport_thm}}

In this section we apply the foregoing results to obtain a generalized
formulation of the transport theorem for a region of any dimension,
and in particular, the surface transport theorem. In view of the postulates
for a motion described in Section \ref{sec:Lipschitz_motion}, we
recall that for an $r$-current of finite mass, $T\in F_{r,\compm}\left(\body\right)$,
and the map $\conf:\ti\times\body\to\reals^{n}$ associated with a
motion by Equation (\ref{eq:Lips_map_motion}), one has 
\[
\conf_{t\#}\left(T\right)\in F_{r,\compm'}\left(\reals^{n}\right),\quad\text{and}\quad\frac{d}{dt}\left(\conf_{t\#}\left(T\right)\right)\in S_{r,\compm'}\left(\reals^{n}\right),\quad\text{for all }t\in\ti.
\]

The standard transport theorem is concerned with the integration of
a density of some extensive property over an evolving region. Here,
the evolving region is generalized and is represented by the evolving
flat chain $\conf_{t\#}\left(T\right)$, and the extensive property
at any instant is represented by a sharp cochain. As any sharp cochain
is a flat cochain, the integration operation in the classical formulation
of the theorem is replaced by the action of a cochain on a chain.
Thus, an evolving extensive property, $\psi$, is viewed as a continuous
curve in the space of sharp $r$-cochains, that is, a continuous mapping
\[
\cochain_{\psi}:\ti\to\left[S_{r}\left(\reals^{n}\right)\right]^{*}.
\]
Moreover, we assume that the curve is differentiable in the topology
of flat $r$-cochains in $\reals^{n}$. That is, the limit 
\[
\lim_{\eps\to0}\frac{\cochain_{\psi}(\tau+\varepsilon)-\cochain_{\psi}(\tau)}{\eps},
\]
exists as a flat $r$-cochain, and it will denoted by $\dot{\cochain}_{\psi}(\tau)$.

For each $t\in\ti$ the total of the extensive property in the flat
$r$-chain $T$ is therefore 
\[
\cochain_{\psi}\left(\conf_{\#}\left(T\right)\right)(t)=\cochain_{\psi}(t)\left(\conf_{t\#}\left(T\right)\right).
\]
For the time derivative of $\cochain_{\psi}\left(\conf_{\#}\left(T\right)\right)$,
it easily follows that we have
\[
\begin{split}\frac{d}{dt}\cochain_{\psi}\left(\conf_{\#}\left(T\right)\right) & =\lim_{\varepsilon\to0}\frac{\cochain_{\psi}(t+\varepsilon)\left(\conf_{t+\epsilon\#}\left(T\right)\right)-\cochain_{\psi}(t)\left(\conf_{t\#}\left(T\right)\right)}{\varepsilon},\\
 & =\cochain_{\psi}(t)\left(\frac{d}{dt}\conf_{t\#}\left(T\right)\right)+\frac{d\cochain_{\psi}}{dt}(t)(\conf_{t\#}\left(T\right)).
\end{split}
\]
By Equation (\ref{eq:ddt(conf(T))}), it follows that
\begin{multline}
\frac{d}{dt}\left(\cochain_{\psi}\left(\conf_{\#}\left(T\right)\right)\right)\mid_{t=\tau}\\
\begin{split} & =\dot{\cochain}_{\psi}(\tau)\left(\conf_{\tau\#}\left(T\right)\right)+\cochain_{\psi}(\tau)\left(\bnd\left(\ve_{\tau}\wedge\conf_{\tau\#}\left(T\right)\right)+\ve_{\tau}\wedge\conf_{\tau\#}\left(\bnd T\right)\right).\end{split}
\label{eq:Transport}
\end{multline}
In the case where $\cochain_{\psi}$ is represented by a smooth differential
$r$-form $\fform{\psi}$ and $\conf$ is a smooth embedding we obtain
\[
\begin{split}\frac{d}{dt}\left(\cochain_{\psi}(t)\left(\conf_{t\#}\left(T\right)\right)\right)\mid_{t=\tau} & =\left(\conf_{\tau\#}\left(T\right)\right)\left(\dot{\fform{\psi}}+\cbnd\fform{\psi}\rest\ve_{\tau}+\cbnd\left(\fform{\psi}\rest\ve_{\tau}\right)\right),\\
 & =\conf_{\tau\#}\left(T\right)\left(\dot{\fform{\psi}}+\Lie_{\ve_{\tau}}\fform{\psi}\right),
\end{split}
\]
which is a generalization of \cite[equation (T0)]{Betounes1986}.
For a time dependent $r$-form $\fform{\psi}$, using \cite[pp.~367]{Marsden1988}
and the notation of Section \ref{sec:Lipschitz_motion}, it follows
that
\begin{eqnarray*}
\frac{d}{dt}\conf_{t}^{\#}\left(\fform{\psi}(t)\right)\mid_{t=\tau} & = & \conf_{\tau}^{\#}\left(\dot{\fform{\psi}}(\tau)+\cbnd\fform{\psi}(\tau)\rest\ve_{\tau}+\cbnd\left(\fform{\psi}(\tau)\rest\ve_{\tau}\right)\right).
\end{eqnarray*}
Applying the last expression to Equation (\ref{eq:Transport}), we
obtain the following relation between the Eulerian and Lagrangian
formulations of the transport theorem 
\[
\frac{d}{dt}\left(\cochain_{\psi}\left(\conf_{\#}\left(T\right)\right)(t)\right)\mid_{t=\tau}=\left(\frac{d}{dt}\conf_{t}^{\#}\left(\cochain_{\psi}(t)\right)\mid_{t=\tau}\right)T.
\]

For the general case, the sharp $r$-cochain $\cochain_{\psi}(\tau)$
is represented by a sharp $r$-form $\fform{\psi}(\tau)$ and the
sharp $r$-cochain $\dot{\cochain}_{\psi}(\tau)$ is represented by
the flat $r$-form $\dot{\fform{\psi}}(\tau)$. As $T$ is a flat
chain of finite mass it may be represented by $\lusb^{n}$-integrable
vector field $T=\eta\wedge\lusb^{n}$. It is concluded that
\begin{multline*}
\frac{d}{dt}\left(\cochain_{\psi}\left(\conf_{\#}\left(T\right)\right)\right)\mid_{t=\tau}\\
\begin{split} & =\left(\conf_{\tau\#}\left(T\right)\right)\left(\dot{\fform{\psi}}(\tau)+\cbnd\fform{\psi}(\tau)\rest\ve_{\tau}+\cbnd\left(\fform{\psi}(\tau)\rest\ve_{\tau}\right)\right),\\
 & =\int_{\body}\conf_{\tau}^{\#}\left(\dot{\fform{\psi}}(\tau)+\cbnd\fform{\psi}(\tau)\rest\ve_{\tau}+\cbnd\left(\fform{\psi}(\tau)\rest\ve_{\tau}\right)\right)(\eta)d\lusb_{x}^{n}.
\end{split}
\end{multline*}

As an example consider the situation where the flat form $\fform{\psi}$
satisfies Cauchy's postulates. Then, associated with the property
are a time dependent flat $r$-form $\fform{\form}$ representing
the source, and a time dependent flat $(r-1)$-form $\fform{\xi}$
representing the flux, such that the differential balance equation
for the property is (see \cite{Segev2012})
\begin{equation}
\dot{\fform{\psi}}+\cbnd\fform{\xi}=\fform{\form}.
\end{equation}
Each of the forms above represents a flat cochain denoted by $\cochain_{\omega},\:\cochain_{\xi},\;\cochain_{\form}$,
respectively, such that 
\[
\dot{\cochain}_{\omega}+\cbnd\cochain_{\xi}=\cochain_{\form}.
\]
 Thus, the transport formula assumes the form
\begin{multline*}
\frac{d}{dt}\left(\cochain_{\psi}\left(\conf_{\#}\left(T\right)\right)\right)\mid_{t=\tau}\\
\begin{split}= & \dot{\cochain}_{\psi}(\tau)\left(\conf_{\tau\#}\left(T\right)\right)+\cochain_{\psi}(\tau)\left(\frac{\partial}{\partial t}\conf_{t\#}\left(T\right)\mid_{t=\tau}\right),\\
= & \dot{\cochain}_{\psi}(\tau)\left(\conf_{\tau\#}\left(T\right)\right)+\cochain_{\psi}(\tau)\left(\bnd\left(\ve_{\tau}\wedge\conf_{\tau\#}T\right)+\ve_{\tau}\wedge\bnd\left(\conf_{\tau\#}T\right)\right),\\
= & \left(\cochain_{\form}\right)\left(\conf_{\#}\left(T\right)\right)(\tau)+\left(\cochain_{\psi}\rest\ve-\cochain_{\xi}\right)\left(\conf_{\#}\left(\bnd T\right)\right)(\tau)+\cochain_{\psi}\left(\bnd\left(\ve\wedge\conf_{\#}T\right)\right)(\tau).
\end{split}
\end{multline*}

\bigskip{}

\noindent \textbf{\textit{Acknowledgments.}} This work was partially
supported by the Perlstone Center for Aeronautical Engineering Studies,
the Kreitman Post-Doctoral Scholarship and the H.~Greenhill Chair
for Theoretical and Applied Mechanics at Ben-Gurion University of
the Negev. The authors wish to thank M. Silhavy for the discussions
and comments he has made on earlier versions of this work.

\bibliographystyle{plain}

\end{document}